\documentclass[journal]{IEEEtran}

\usepackage{amsmath}
\usepackage{abraces}
\usepackage{amsfonts}
\DeclareMathOperator*{\argmin}{arg\,min}
\usepackage{amsthm,amssymb}
\usepackage{algorithm,algorithmic}
\usepackage{amsmath,bm}
\usepackage[english]{babel}
\usepackage{optidef}
\usepackage[inline]{enumitem}
\usepackage{graphicx,cite}
\usepackage{sidecap}
\usepackage{hyperref}
\usepackage{subfig}
\usepackage[euler]{textgreek}
\usepackage{float}
\usepackage{xcolor}
\usepackage{mysymbol}
\usepackage{nicefrac}      
\usepackage{microtype}      
\usepackage{lipsum}

\newtheorem{theorem}{Theorem}
\newtheorem{corollary}{Corollary}
\newtheorem{lemma}{Lemma}

\def\Hw{\mathbf{H}(\boldsymbol{\omega})}
\def\xhatw{\mathbf{\hat{x}}(\boldsymbol{\omega})}
\def\bw{\boldsymbol{\omega}}
\def\bw{\boldsymbol{\omega}}
\def\diagw{\text{diag}(\bw)}

\title{Node-Adaptive Regularization for \\ Graph Signal Reconstruction}

\author{Maosheng Yang, Mario Coutino, Geert Leus and Elvin Isufi\thanks{The authors are with Faculty of Electrical Engineering, Mathematics \& Computer Science, Delft University of Technology, Delft, The Netherlands. Part of this work is presented in \cite{yang2020}.
This research is supported in part by the ASPIRE
project (project 14926 within the STW OTP programme), financed by the Netherlands Organization for Scientific Research (NWO). Mario Coutino is partially supported by CONACYT. Maosheng Yang's work is supported by the Microelectronic-M.Sc. track scholarship from Delft University of Technology.}}

\begin{document}

\maketitle
\begin{abstract}
A critical task in graph signal processing is to estimate the true signal from noisy observations over a subset of nodes, also known as the reconstruction problem. In this paper, we propose a node-adaptive regularization for graph signal reconstruction, which surmounts the conventional Tikhonov regularization, giving rise to more degrees of freedom; hence, an improved performance. We formulate the node-adaptive graph signal denoising problem, study its bias-variance trade-off, and identify conditions under which a lower mean squared error and variance can be obtained with respect to Tikhonov regularization. Compared with existing approaches, the node-adaptive regularization enjoys more general priors on the local signal variation, which can be obtained by optimally designing the regularization weights based on Prony's method or semidefinite programming. As these approaches require additional prior knowledge, we also propose a \textit{minimax} (worst-case) strategy to address instances where this extra information is unavailable. Numerical experiments with synthetic and real data corroborate the proposed regularization strategy for graph signal denoising and interpolation, and show its improved performance compared with competing alternatives.
\end{abstract}
\section{Introduction} \label{sec1}
Graphs, as models to represent data, do not only capture information about entities (nodes) that comprise them, but also encode the interactions between these entities (edges). Graphs are promising to deal with high-dimensional data arising in many applications such as social media, sensor networks, and transportation networks, to name a few \cite{huang2018, russell2011, mehri2012, emerging_gsp, deshpande2004, lee2015, scroggie2000 }. To analyze such networked data, tools from signal processing have been extended in the area of \textit{graph signal processing} (GSP) \cite{emerging_gsp},\cite{gsp_overview}.

As in classical signal processing, the task of estimating the underlying signal from noisy observations, is critical in GSP as well. Over the past few years, a large amount of research has been focused on this topic~\cite{emerging_gsp,gsp_overview,local_set_based,localized_iterative,distr_tracking,adaptive_sampling,wiener,TV_primal_dual,cheb,biasvar,krr1,mul_ker2,mul_kernel}. Within this body of work, most approaches solve a least squares problem penalized by different regularizers. The most commonly penalization method is the one based on the so-called Tikhonov regularization \cite{groetsch1984}. It recovers the graph signal by penalizing the data fitting term with the graph Laplacian quadratic form, also known as graph Laplacian regularizer (GLR) \cite{chung1997, biasvar}. This GLR based method and its variants have succeeded in many applications such as image smoothing, point cloud denoising, semi-supervised learning, data classification and graph signal denoising \cite{emerging_gsp, milanfar2012, dspog:freq,tomasi1998, liu2016, pang2017,dinesh2020}.\footnote{Note that in \cite{tomasi1998, pang2017,dinesh2020}, the graph edge weights (the graph Laplacian) are not fixed but (designed) updated in a specific way for tasks like image and point could denoising. This edge weights updating idea seem to be related to the idea introduced in our paper, but they are specially designed for/based on the given tasks, for instance by exploiting the properties of images and point clouds. This is conceptually different from the idea in this paper.}

Other regularizers have also shown their applicability in different scenarios. For example, the definition of total variation  of discrete signals has been extended to graphs, i.e., the graph signal total variation notion based on the adjacency matrix, graph Laplacian and their variants \cite{tibshirani2005, mallat1999, dspog:freq, couprie2013}. In \cite{wang2016trend}, it is used as a graph signal regularizer to perform trend filtering on graphs, effectively extending the generalized Lasso problem~\cite{tibshirani2011} to graphs. By generalizing total variation to the graph setting, we can consider higher-order graph signal differences across the nodes and performs graph signal reconstruction in instances where the graph signal is not necessarily globally smooth. Similarly, in \cite{krr1,mul_ker2,mul_kernel}, the notion of graph Laplacian regularization is extended by means of graph kernels such as the diffusion kernel, random walk based kernels and graph bandlimited kernels. As these regularizers tend to increase the computational complexity of the reconstruction task, several works have been focused on reducing of the complexity for finding regularized solutions by means of iterative and distributed implementations, such as adaptive graph signal estimation, primal-dual gradient method, Chebyshev polynomial approximation and so on \cite{local_set_based,localized_iterative,distr_tracking,adaptive_sampling,wiener,TV_primal_dual,cheb}.

Although state-of-the-art address the graph signal reconstruction, almost all methods adopt a single parameter to control the reconstruction performance; that is, to balance the fitting error with regularization term. For instance, in the recent work~\cite{biasvar}, the authors investigate the bias-variance trade-off for Tikhonov regularizer and propose an optimally design strategy for the regularization parameter which matches the order of the bias and variance term. Despite this effort, it is clear that a scalar regularization parameter is insufficient to impose local penalties over the graph signal. To see this, consider an instance of the Tikhonov denoising that penalizes the error fitting term with signal smoothness. A scalar regularization weight can only penalize in terms of the \emph{global} smoothness, instead of local variability.  Such a local penalty is important, for instance, for signals that are piecewise-smooth or piecewise-constant \cite{chen2016, wang2016trend,tomasi1998,liu2016}, or in some networks where we want to detect anomaly nodes, etc. The global penalty term limits the denoising performance to deal with signals with general characteristics. Thus an increased number of regularization parameters for handling such situations is needed. In an attempt to increase the degrees of freedom (DoFs) of the regularization, in~\cite{TV_primal_dual}, it is proposed to minimize the signal total smoothness by regularizing separately the fitting error of each individual nodal measurement. Despite that this approach can be considered as a multi-parameter based regularization, it only focuses on a measurement-wise regularization and ignores the coupling between the graph signal and the topology.

Due to the need of a multi-parameter regularization, which considers the connectivity of the graph, in this paper, we propose a \emph{node-adaptive (NA) regularizer} to increase the DoFs by applying node-dependent weights to fine-tune the trade-off between the fitting error and regularization term. With these enhanced DoFs, we expect to achieve a better reconstruction performance without affecting the method complexity. For the sake of exposition, we solely focus on the \emph{NA Tikhonov-based reconstruction} to make direct comparisons with earlier works, e.g.,~\cite{biasvar}, however, our findings can be generalized with the adequate changes to other regularization penalties, e.g., a graph shift operator based ridge regression penalty \cite{mul_ker2,mul_kernel}. 

\subsection{Contributions}
To be more specific, the main questions we address in this work are: (i) \emph{how the bias-variance trade-off behaves for the NA Tikhonov regularization?} and (ii) \emph{how to design the NA weights optimally?} Aimed to give answers to these questions, we make the following main contributions:
\begin{itemize}
\item[1)] We formulate the NA Tikhonov regularization problem under a deterministic signal model assumption, derive its solution in closed-form, and study the respective bias-variance trade-off. 
\item[2)] We derive the conditions for the NA weights to allow a smaller mean-squared error (MSE) and variance compared to traditional Tikhonov regularization.
\item[3)] To design the NA weights, we propose three methods based on  minimizing the MSE. The first two methods leverage  Prony's method from classical signal processing and convex relaxation techniques. The third method uses a \textit{minimax} strategy to design the weights in a worst-case setting. The latter addresses scenarios where only upper- and lower-bounds of graph signals are available.
\item[4)] We corroborate the theoretical findings of this work, using both synthetic and real-world data, and show that the proposed NA denoising performs well with respect to the competing alternatives, especially in low signal-to-noise ratio (SNR) settings.
\end{itemize}

\subsection{Outline and Notation}
The rest of this paper is structured as follows: Section~\ref{sec2} reviews the conventional Tikhonov regularizer. Section~\ref{sec3} formulates the NA Tikhonov regularization problem and studies its bias-variance trade-off. Section \ref{sec4} presents the proposed optimal designs for the NA regularizer weights. Section~\ref{sec5} contains numerical results to validate the theoretical findings. Finally, Section~\ref{sec6} concludes the paper.

\textbf{Notation:} Scalars, vectors, matrices and sets are denoted by lowercase letters ($x$), lowercase bold letters ($\mathbf{x}$), uppercase bold letters ($\mathbf{X}$), and calligraphic letters ($\mathcal{X}$), respectively. $x_i$ represents the $i$-th entry of the vector $\bf{x}$, and $X_{ij}$ denotes the $(i,j)$th element of the matrix $\mathbf{X}$. $\mathbf{X}^\top$ and $\mathbf{X}^{-1}$ are the transpose and inverse of the matrix $\mathbf{X}$. $\mathbf{1}$ and $\mathbf{I}$ are the all-one vector and identity matrix. $\text{diag}(\cdot)$ is a diagonal matrix with its arguments on the main diagonal. $\lVert\cdot\lVert_2$ is the Euclidean norm.
We use $\mathbb{R}$ to denote the set of real numbers and $\mathcal{S}_+^{N\times N}$ to denote the set of $N\times N$ positive semi-definite matrices.
$\mathbb{E}[\cdot]$ is the expectation operator, $\tr(\cdot)$ is the trace operator  and $\text{supp}(\mathbf{X})$ is the support of $\mathbf{X}$.  If $\mathbf{x}$ is a random variable, then its covariance matrix is $\text{cov}(\mathbf{x})=\mathbb{E}[(\mathbf{x}-\mathbb{E}[\mathbf{x}]) (\mathbf{x}-\mathbb{E}[\mathbf{x}])^\top]$. 
\section{Node-invariant regularizer} \label{sec2}
Consider an \emph{undirected} graph $\mathcal{G}=\left( \mathcal{V},\mathcal{E}\right)$, where $\mathcal{V}=\{1,\dots,N\}$ is the set of $N$ nodes and $\mathcal{E}$ the set of $M$ edges such that if nodes $i$ and $j$ are connected, then $\left(i,j \right) \in \mathcal{E}$. The neighborhood set of node $i$ is $\mathcal{N}_i=\left\lbrace j|\left(i,j \right)\in \mathcal{E}  \right\rbrace $. The graph can be represented by its adjacency matrix $\bbA$ with entry $A_{ij}\geq0$ if $\left(i,j \right) \in \mathcal{E}$ and $A_{ij}=0$, otherwise. Alternatively, the graph can also be represented by its graph Laplacian matrix
$\mathbf{L=\text{diag}(\mathbf{A1})-A}$. 

On the vertices of $\ccalG$, we define a graph signal $\mathbf{x}~=~[x_1,\dots,x_N]^\top$ whose $i$-th element $x_i$ is the signal value on node $i$. The graph signal smoothness can be measured by the \emph{graph Laplacian quadratic form}
\begin{equation}\label{sigvar}
	S_2(\mathbf{x}) = \frac{1}{2}\sum_{i\in\mathcal{V}} \sum_{j\in\mathcal{N}_i}A_{ij}(x_i-x_j)^2 = \mathbf{x}^\top\mathbf{Lx}.
\end{equation}
A signal $\bbx$ is said to be \emph{smooth} over $\ccalG$, if $S_2(\mathbf{x})$ is small \cite{emerging_gsp}. This measure is commonly used as a regularizer to recover smooth graph signals from noisy measurements \cite{emerging_gsp,gsp_overview,biasvar}.

%

Consider the measurements $\bby = \bbx^* + \bbn$, where $\mathbf{x}^*\in \mathbb{R}^{N}$ is a smooth graph signal and $\mathbf{n}\in\mathbb{R}^{N}$ is zero-mean noise with covariance matrix $\boldsymbol{\Sigma}$.
Assuming $\mathbf{x}^*$ is smooth over the graph, we can recover it by solving the Tikhonov regularized problem

\begin{equation}\label{eq.regls}
	\mathbf{\hat{x}}(\omega_0) = \underset{\mathbf{x}\in\mathbb{R}^N}{\argmin}\,\,\, \lVert\mathbf{y-x}\lVert^2_2 \,+ \,\omega_0 \mathbf{x}^\top\mathbf{Lx}
\end{equation}
where scalar $\omega_0>0$ is the regularization parameter. The first term in \eqref{eq.regls} forces the estimate to be close to the observed signal (fitting term), while the second term promotes global signal smoothness. The trade-off between these two quantities is governed by the scalar weight $\omega_0$. The closed-formed solution of \eqref{eq.regls} is 
\begin{equation}\label{eq.reglssol}
	\mathbf{\hat{x}}(\omega_0) = \left( \mathbf{I}+\omega_0\mathbf{L} \right) ^{-1}\mathbf{y}:=\mathbf{H}(\omega_0)\mathbf{y}
\end{equation}
where we defined the graph filter $\mathbf{H}(\omega_0)\triangleq\left( \mathbf{I}+\omega_0\mathbf{L} \right) ^{-1}$ \cite{arma}. 

\def\b{\textbf{b}}
\def\bx{\mathbf{x}}
The smooth regularizer in~\eqref{eq.regls} biases the estimator $\hbx(\omega_0)$ in~\eqref{eq.reglssol}. The bias \b$(\omega_0)$ is given by
\begin{equation}\label{bias_w0}
\begin{split}
        \b(\omega_0) &=  \mathbb{E}[\mathbf{\hat{x}}(\omega_0)]-\mathbf{x}^* = 
        (\mathbf{H}(\omega_0)-\mathbf{I})\mathbf{x}^*
\end{split}
\end{equation}
which is controlled by $\omega_0$. The variance is also controlled by $\omega_0$ and has the form

\begin{align}\label{var_w0}
    \text{var}(\omega_0) &= \mathbb{E} [ \lVert \hat{\bx}(\omega_0) - \mathbb{E}(\hat{\bx}(\omega_0)) \lVert^2]
    = \tr(\mathbf{H}^2(\omega_0)\boldsymbol\Sigma).
\end{align}

By combining the bias and the variance, we can quantify the performance of the estimator in \eqref{eq.reglssol} through its MSE
\begin{equation}\label{mse_w0}
\begin{split}
\text{mse}(\omega_0) & = \mathbb{E}[ \lVert \mathbf{\hat{x}}(\omega_0) - \mathbf{x^*} \lVert_2^2 ]=  \lVert \b(\omega_0) \lVert_2^2 + \text{var}(\omega_0) \\
& = \tr(\left( \mathbf{I-H}(\omega_0)\right) ^2\mathbf{x^*x}^{*\top}) + \tr( \mathbf{H}^2(\omega_0)\boldsymbol{\Sigma}).
\end{split}	
\end{equation}
The MSE shows the bias-variance trade-off imposed by the smoothness regularizer. If  scalar $\omega_0$ is reduced, we achieve a lower bias but a higher variance, and vice-versa.

Ways to select the parameter $\omega_0$ can be devised using, e.g., the discrepancy principle \cite{discrepancy_1993, Anzengruber_2009}, the $L$-curve criterion \cite{L-curve_hansen} and the generalized cross-validation \cite{GCV_gene}. In GSP, a natural optimal parameter selection criterion is based on the minimization of the MSE, which is investigated in \cite[Thm. 3]{biasvar}, where parameter $\omega_0$ is optimally found by matching the order of the bias contribution and the variance. %

The regularizer in \eqref{sigvar} is a global graph signal measure, i.e., the signal smoothness over the entire graph. And the impact of this global regularizer in~\eqref{eq.regls} is controlled by the scalar $\omega_0$. We can differently think of $\omega_0$ as a common coefficient that equally pre-weights the signal on all nodes when computing the smoothness measure, i.e., $\sqrt{\omega_0}\bx^\top\mathbf{L}\sqrt{\omega_0}\bx$. Due to this characteristic, we refer to problem~\eqref{eq.regls} as a node-invariant (NI) regularization problem.
If the signal is globally smooth over the underlying graph, the NI regularization will result in a good signal reconstruction performance. However, as it only focuses on global behaviours, ignoring local signal details, it might lead to unsatisfactory recovery performance when the graph signal is not globally smooth and/or when the noise level is different at each node.

Although the graph signal is not globally smooth in most situations, it presents a smooth behaviour in different regions, e.g., smoothness over local details. This prior generalizes the global smoothness and is amenable for learning. To enhance the role of local signal detail and improve the estimator's MSE, we propose a graph signal regularization strategy that substitutes the global penalty term $S_2(\bx)$ with a local penalty on each node. The proposed approach, named node-adaptive (NA) regularization, allows each node $i$ to weight its signal $x_i$ with an individual scalar $\omega_i$. The enhanced DoFs of the NA regularizer can in turn improve the bias-variance trade-off and reduce the overall MSE. 
To understand the NA regularizer, we first conduct a detailed bias-variance trade-off analysis and then propose a design strategy for the node-adaptive weights $\omega_1, \ldots, \omega_N$ to optimize such a trade-off. Therefore, our goal in the upcoming sections is to formalize the NA regularization problem, analyze its statistical properties, and develop optimal weight design strategies that minimize the MSE.
 
\section{Node-adaptive regularizer} \label{sec3}

Consider a vector of parameters $\boldsymbol{\omega}=[\omega_1,\dots,\omega_N]^\top \in \reals^N$ and define the \emph{node-adaptive Laplacian operator}
\def\Sw{\mathbf{S}(\boldsymbol{\omega})}
\begin{equation}\label{eq:s}
   \mathbf{S}(\boldsymbol{\omega})\triangleq\text{diag}(\boldsymbol{\omega})\mathbf{L}\text{diag}(\boldsymbol{\omega}) = \boldsymbol{\omega\omega}^\top\odot\mathbf{L}
\end{equation}
where $\odot$ is the element-wise Hadamard product. Note that for any $\bbomega$, the parametric shift operator matrix $\mathbf{S}(\boldsymbol{\omega})$ is positive semi-definite (see Lemma~\ref{lemma4} in Appendix~\ref{appendix.A}) and has the same support as the graph Laplacian $\bbL$ --properties that will result useful in the sequel. The parameterized graph Laplacian quadratic form for \eqref{eq:s} has the form [cf.\eqref{sigvar}]
\begin{equation}\label{eq.naQform}
    \begin{split}
        S_2(\bx,\bw) &= \mathbf{x}^\top\mathbf{S}(\bw)\mathbf{x} = (\diagw \mathbf{x})^\top \mathbf{L}(\diagw \mathbf{x}) \\ &= \frac{1}{2}\sum_{i\in\mathcal{V}} \sum_{j\in\mathcal{N}_i}A_{ij}(\omega_i x_i-\omega_j x_j)^2.
    \end{split}
\end{equation}
This quadratic form can be seen as, first pre-weighting each entry $x_i$ of the signal $\bbx$ with a parameter $w_i$ and then computing the regular quadratic measure in \eqref{sigvar} w.r.t. Laplacian $\bbL$. We can now use \eqref{eq.naQform} to recover a graph signal $\bbx^*$ from the noisy measurements $\bby$ by solving the convex problem
\begin{equation}\label{eq.nvls}
\begin{split}
	 \mathbf{\hat{x}}(\boldsymbol{\omega}) 
	 = \argmin_{\mathbf{x}\in\mathbb{R}^N} \,\,\, & \lVert\mathbf{y-x}\lVert^2_2 \,+ \, S_2(\bx,\bw) ,
\end{split}
\end{equation}
where the trade-off between the fitting error and the regularization term is now controlled by the $N$ parameters in $\bbomega$. We name \eqref{eq.nvls} as \emph{node-adaptive graph signal regularization}. We give the following comments on this regularization framework.
\begin{itemize}
    \item When $\bw = \omega_0\bm 1$, NA graph signal regularization particularizes to the NI case in \eqref{eq.regls}.
    \item NA regularization is not restricted to the form in \eqref{eq.nvls}. The same generalization can be applied when variants of the graph Laplacian are used in the regularizer. For example, if we use a random walk graph Laplacian in \eqref{eq.nvls}, the GLR used in \cite{liu2016} is a special case.
    \item We can see that NA graph signal regularization uses a parametric graph Laplacian of form \eqref{eq:s} where the edge weights are adapted by $\bw\bw^\top$ which will be optimally designed in the sense of minimizing MSE. In this way, our work acts as an alternative or a generalization to the proposed GLR in image and point could denoising \cite{tomasi1998, dinesh2020} where the graph edge weights are (designed) updated in a specific task-driven way, for instance, as a function of (feature) difference between two connected nodes.
\end{itemize}

To see the impact on local graph signal details in problem \eqref{eq.nvls}, consider a simple piecewise smooth graph signal with slowly varying nonzero values  in a node subset $\ccalV^\prime \subseteq \ccalV$ and zero on the remaining nodes in $\ccalV\backslash{\ccalV^\prime}$. We can then weight the nodes in $\ccalV^\prime$ with a common scalar, and another different scalar for the remaining nodes, since the signal is not globally smooth but locally in two separate parts. Problem \eqref{eq.nvls} will seek for a graph signal that is locally smooth on $\ccalV^\prime$ and on $\ccalV\backslash \ccalV^\prime$. Instead, problem \eqref{eq.regls} will weight every node with $\omega_0$, effectively looking for a graph signal that is smooth over all nodes and ignores the local detail. This simplified example shows how the NA regularization generalizes the NI case by considering a local smoothness penalty and relates to the piecewise-constant, -smooth notions. 
In the sequel, we design parameter $\bw$ to learn the local signal smoothness priors.

The optimal analytical solution for problem \eqref{eq.nvls} can be found by setting its gradient to zero, i.e.,
\begin{equation}\label{eq.xhat}
\mathbf{\hat{x}}(\boldsymbol{\omega}) = \left( \mathbf{I}+\mathbf{S}(\boldsymbol{\omega}) \right) ^{-1} \mathbf{y}:=\bbH(\bbomega)\mathbf{y}
\end{equation} 
where we have defined the NA filter\footnote{Here the definition of NA filter is different from the node variant graph filter in \cite{segarra2017}} $ \bbH(\bbomega) \triangleq 
\left( \mathbf{I}+\mathbf{S}(\boldsymbol{\omega}) \right)^{-1}$, which is PSD by definition. Despite the similarity with \eqref{eq.reglssol}, the optimal solution in \eqref{eq.xhat} is now governed by the vector of parameters $\bbomega$. This vector changes the bias-variance trade-off as we discuss in the next section.

\begin{remark}
The optimal regularized solutions \eqref{eq.reglssol} and \eqref{eq.xhat} can also be interpreted as graph filtering operations \cite{emerging_gsp}. In particular, while~\eqref{eq.reglssol} filters the measurements $\bby$ with an autoregressive graph filter with denominator coefficients $(1; \omega_0)$ which is common for all nodes \cite{arma}, the node-adaptive expression~\eqref{eq.xhat} filters $\bby$ with an autoregressive edge-varying filter with edge varying coefficients $(\bbI; \bbomega \bbomega^\top)$ \cite{advance}.
\end{remark}

\subsection{Bias-variance trade-off} \label{sec3.3}
Similar to \eqref{mse_w0}, the MSE for estimator \eqref{eq.xhat} has the form 
\begin{equation}\label{mse_w}
\begin{split}
\text{mse}(\boldsymbol{\omega}) 
& =  \lVert \b(\boldsymbol{\omega}) \lVert_2^2 + \text{var}(\bw)  \\
&= \tr (( \mathbf{I}-\Hw ) ^2\mathbf{x^*x}^{*\top})+ \tr(\Hw^2 \boldsymbol{\Sigma})\\
\end{split}	
\end{equation}
with bias 
\begin{equation}\label{bias_w}
\begin{split}
        \b(\boldsymbol{\omega}) = ( (\mathbf{I}+\diagw\mathbf{L}\diagw)^{-1}-\mathbf{I} ) \mathbf{x}^*
\end{split}
\end{equation}
and variance 
\begin{equation}\label{var_w}
\begin{split}
    \text{var}(\boldsymbol{\omega}) = \tr ( (\mathbf{I}+\diagw\mathbf{L}\diagw)^{-2}\boldsymbol{\Sigma} ).
\end{split}
\end{equation}
As it follows from \eqref{mse_w}-\eqref{var_w}, the bias-variance trade-off is now controlled by $\bbomega$. If all entries of $\bbomega$ are close to zero the bias is low and the MSE is governed by a high variance. If all entries of $\bbomega$ are far from zero, the bias is large and governs the MSE, achieving a small variance. However, the enhanced DoFs of the NA regularizer compared to the NI one allow us to identify an interval for $\bbomega$ that guarantees a smaller reconstruction variance while maintaining a lower MSE.

\begin{lemma}\label{le.res}
	Consider the node-invariant and the node-adaptive estimates $\mathbf{\hat{x}}(\omega_0)$ and $\xhatw$ as  in~\eqref{eq.reglssol} and~\eqref{eq.xhat}, respectively. Consider also the respective variances over all nodes $\text{var}(\omega_0)$ in~\eqref{var_w0} and $ \text{var}(\boldsymbol{\omega})$ in~\eqref{var_w}. If all node-adaptive weights $\bw = [\omega_1,\dots,\omega_N]^\top$ satisfy
	\begin{equation}\label{eq.condProp1}
	\omega_0 \le \omega_i^2, \quad \text{ for } i = 1,2,\dots,N
	\end{equation}
     then $\text{var}(\boldsymbol{\omega}) \leq \text{var}(\omega_0)$.
\end{lemma}
\begin{IEEEproof}
	See Appendix \ref{appendix.B}.
\end{IEEEproof}

While a reduced variance is useful for signal recovery, it often comes at  expenses of an increased bias. To see how sensitive the changes in the two quantities are, we illustrate in \autoref{fig:prop1study} the bias, the variance, and the MSE for the NI and the NA regularizers. We see there exists a region for the NA weights where both the variance and the MSE of the NA regularizer are lower compared with those of the NI regularizer, while the bias is not increased significantly. The following theorem provides sufficient conditions on $\bbomega$ to identify this region. 

\begin{figure}
    \centering
    \includegraphics[width=1\linewidth]{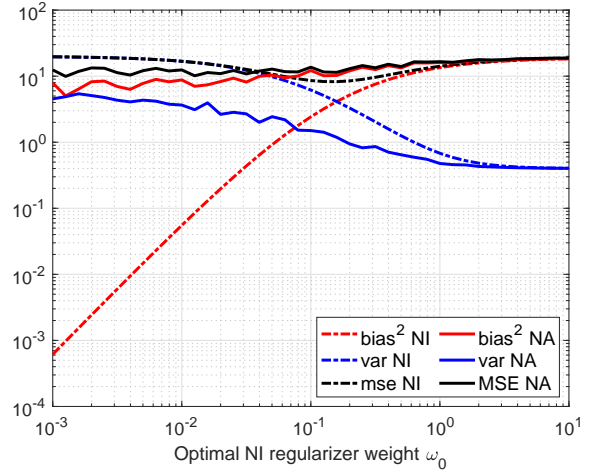}
    \caption{Bias-variance trade-off for recovering a graph signal with the NA and NI regularizer over an Erdos-Reny\`i graph, $\text{SNR}=0\text{ dB}$. The detailed settings are in Section~\ref{sec5.1}.  The node adaptive weights are chosen randomly to satisfy the result in Lemma~\ref{le.res}.}
    \label{fig:prop1study}
\end{figure}

\begin{theorem}\label{th.mse}
Consider the measurements $\mathbf{y} = \bbx^* + \bbn$ with desired graph signal $\bx^*$ and noise~$\mathbf{n}\sim\mathcal{N}(\mathbf{0},\boldsymbol{\Sigma})$. Let $\bbL$ be the graph Laplacian with maximum eigenvalue $\lambda_{\max}(\mathbf{L})$. Further, define a rank-one matrix $\mathbf{P}:=\mathbf{x^*x}^{*\top}\boldsymbol{\Sigma}^{-1}$ and let $\rho$ be its only non-zero eigenvalue. Define also a rank-one matrix $\boldsymbol{\Gamma}=\mathbf{P(I+P)}^{-1}$ and let $\gamma = \rho(1+\rho)^{-1} \in (0,1)$ be its only non-zero eigenvalue. Consider also the mean squared error of the node-invariant estimate $\mathbf{\hat{x}}(\omega_0)$ in~\eqref{eq.reglssol} and node-adaptive estimate $\mathbf{\hat{x}}(\boldsymbol{\omega})$ in~\eqref{eq.xhat}.
Then, if 
\begin{subequations}
\begin{align}
& \omega_0 \leq \omega_i^2, \text{ for } i = 1,2,\dots,N \label{eq.cond1}\\
& 2\gamma \leq \frac{1}{1+\omega_0\lambda_{\text{max}}(\mathbf{L})} + \frac{1}{1+\text{max}\left\lbrace \omega_i^2\right\rbrace \lambda_{\text{max}}(\mathbf{L})} \label{eq.cond2}
\end{align}
\end{subequations}
both the variance and the mean squared error of the node-adaptive regularizer are smaller than those of the node-invariant one; i.e., $\text{var}(\boldsymbol{\omega}) \leq \text{var}(\omega_0)$ and $\text{mse}(\boldsymbol{\omega}) \leq \text{mse}(\omega_0)$.
\end{theorem}
\begin{IEEEproof}
	See Appendix. \ref{appendix.C}.
\end{IEEEproof}

Condition~\eqref{eq.cond2} is easier satisfied when the eigenvalue $\gamma~\rightarrow~0$, i.e., the signal-to-noise ratio (SNR) is low, or $\rho \rightarrow 0$. This indicates that the NA is more powerful in harsher scenarios. In contrast, when  $\gamma \rightarrow 1$ and, thus $\rho \rightarrow \infty$, i.e., the SNR is high, the condition for the NA regularization to outperform NI one is hard to satisfy.

\begin{corollary}\label{cor1}
    Under the same settings of Theorem~\ref{th.mse}, the condition
	\begin{equation}\label{eq.corCond}
	    \text{max}\left\lbrace \omega_i^2 \right\rbrace  \leq (\rho\lambda_{\text{max}}(\mathbf{L}))^{-1}, \text{ for }i = 1,2,\dots,N 
	\end{equation}
	 guarantees condition \eqref{eq.cond2} is satisfied.
\end{corollary}
\begin{IEEEproof}
	See Appendix. \ref{appendix.D}.
\end{IEEEproof} 

Corollary~\ref{cor1} provides an easier condition for the NA parameters $\bbomega$ compared to Theorem~\ref{th.mse}. In addition, condition~\eqref{eq.corCond} provides a clearer link between $\bbomega$ and the SNR. If $\rho \to 0$, hence $\gamma \to 0$, the noise is high and the upper bound for $\omega_i^2$ increases, implying a larger weight on local smoothness is needed. If $\rho \rightarrow \infty$, hence $\gamma \to 1$, the noise vanishes and the upper bound for $\omega_i^2$ goes to zero, implying the regularization has little effect. This performance analysis the NA regularizer compared to the NI one is supported by our numerical findings in Section~\ref{sec5}.

\subsection{Implementation}

The matrix inversion in  optimal estimate $\hat{\mathbf{x}}(\boldsymbol\omega)$ [cf.~\eqref{eq.xhat}] makes the node-adaptive regularizer  challenging  to be implemented on large graphs or in a distributed manner. Fortunately, in both cases the inverse can be approximated with a cost linear w.r.t. the number of graph edges, i.e., $\ccalO(M)$, by leveraging the graph filtering equivalence of \eqref{eq.xhat}; see, e.g.,~\cite{arma,filter_arma,advance}. The key to such linear cost lies in the sparsity of the parametric shift operator $\bbS(\bbomega)$, which coincides with the sparsity of the graph [cf. \eqref{eq:s}]. Because of this sparsity, the graph signal shifting operation $\bx^{(1)} = \Sw \bx = (\bbomega\bbomega^\top \odot \bbL) \bx $ has a cost of order linear in the number of edges $M$. Moreover, this operation is local over the graph and the $i$-th signal value is given by
\begin{equation}\label{eq.sparseImpl}
     x^{(1)}_i = \omega_i \sum_{j\in\mathcal{N}_i} A_{ij}(\omega_i x_i - \omega_j x_j). 
\end{equation}
By exploring \eqref{eq.sparseImpl}, we detail next how the NA filter can be implemented using the  conjugate gradient method \cite{cg_method} and distributed graph filters \cite{elvinphd,arma}.

\smallskip\noindent\textbf{Centralized.} To implement \eqref{eq.xhat} efficiently, we first rephrase it as a linear system
\def\xhatw{\mathbf{\hat{x}}(\boldsymbol{\omega})}
\begin{equation}\label{central_imp}
   (\mathbf{I}+\Sw) \xhatw = \mathbf{y}
\end{equation}
and then employ conjugate gradient \cite{cg_method} to obtain $\hat{\mathbf{x}}(\boldsymbol\omega)$. For completeness,  Algorithm~\ref{alg-CG} summarizes the required steps. Exploiting the local operation \eqref{eq.sparseImpl} in Steps 7 and 9 of Algorithm~\ref{alg-CG}, which are the main sources of computing exhausts, and running the conjugate gradient method for $T$ iterations, we have a cost of order $\mathcal{O}(TM)$.

\def\xhat{\mathbf{\hat{x}}}
\begin{algorithm}[!t]
	\caption{Conjugate gradient method for solving~\eqref{central_imp}} 
	\begin{algorithmic}[1] \label{alg-CG}
		\STATE \textbf{Input:} $\xhat_{(0)}$, node-adaptive regularizer weights $\bw$, accuracy $\epsilon$, number of iterations $T$\\
		\STATE \textbf{Initialization:} 
		\STATE  $\Sw =  \diagw\mathbf{L\diagw}$ \\
		\STATE  $\mathbf{b}_{(0)}  = \mathbf{r}_{(0)} = \mathbf{y} - (\mathbf{I}+\Sw)\xhat_{(0)}$ \\
		\STATE  $d_{(0)} = d_{new}  = \mathbf{r}_{(0)}^\top\mathbf{r}_{(0)}$\\
		\STATE \textbf{while} {$\tau<T$ and $d_{new} > \epsilon^2 d_{(0)}$}	\\
		\STATE $c_{(\tau)} = \frac{ d_{new} }{ \mathbf{b}_{(\tau)}^{\top} (\mathbf{I}+\Sw)  \mathbf{b}_{(\tau)} }$ \\ 
		\STATE $\xhat_{(\tau+1)} = \xhat_{(\tau)} +  c_{(\tau)}\mathbf{b}_{\tau}$ \\
		\STATE $\mathbf{r}_{(\tau+1)} =  \mathbf{r}_{(\tau)} - c_{(\tau)} (\mathbf{I}+\Sw)  \mathbf{b}_{(\tau)}$  \\ 
		\STATE $d_{old} = d_{new}, \, d_{new} = \mathbf{r}_{(\tau+1)}^{\top}\mathbf{r}_{(\tau+1)}$ \\
		\STATE $\mathbf{b}_{(\tau+1)} = \mathbf{r}_{(\tau+1)} + \frac{d_{new}}{d_{old}} \mathbf{b}_{(\tau)}$ \\
		\STATE $\tau = \tau+1$  \\
		\STATE \textbf{Output:}  $\xhatw = \xhat_{(\tau+1)}$\\
	\end{algorithmic}
\end{algorithm}

\smallskip\noindent\textbf{Distributed.}
To implement \eqref{eq.xhat} distributively with graph filters, we start with a random initialization $\mathbf{\hat{x}}_0$ for estimate $\hbx(\omega)$. At iteration $\tau$, the distributed estimate follows the recursion
\begin{equation}\label{eq.dis_imp1}
\mathbf{\hat{x}}_\tau(\boldsymbol{\omega}) = -\Sw\mathbf{\hat{x}}_{\tau-1}(\boldsymbol{\omega}) + \mathbf{y}.
\end{equation}
From~\eqref{eq.sparseImpl}, we can see the term $\Sw\mathbf{\hat{x}}_{\tau-1}(\boldsymbol{\omega})$ implies nodes communicate with their neighbors and exchange information about the previous estimate $\mathbf{\hat{x}}_{\tau-1}(\boldsymbol{\omega})$, which has a cost of order $\ccalO(M)$.

When $\bbomega$ satisfies the spectral norm inequality $\|\Sw\| < 1$, recursion \eqref{eq.dis_imp1} leads to the steady-state ($\tau \to \infty$) estimate
\begin{equation}
	\begin{split}
	\mathbf{\hat{x}}(\boldsymbol{\omega}) \triangleq \lim\limits_{\tau\rightarrow\infty}\mathbf{\hat{x}}_\tau(\boldsymbol{\omega}) = \sum_{\tau=0}^{\infty}(-\Sw)^\tau\mathbf{y} = (\mathbf{I}+\Sw)^{-1}\mathbf{y},
	\end{split}
\end{equation}
which matches the optimal solution~\eqref{eq.xhat}. Arresting, therefore, \eqref{eq.dis_imp1} in $T$ iterations leads again to a distributed communication and computational cost of order $\ccalO(TM)$.

The sparsity of parametric shift operator $\bbS(\bw)$ leads to a reduced implementation complexity for both centralized and distributed ways, which is critical in practice. In the following, we investigate how to optimally design the NA weights so to obtain the optimal performance. 

\section{Weight design} \label{sec4}
\def\x*{\mathbf{x}^*}
\def\minmax{\textit{min-max }}
\def\supp{\text{supp}}
\def\bW{\boldsymbol{\Omega}}
In this section, we focus on designing the NA weights $\bw$ in a minimum MSE (MMSE) sense. The estimation error between the optimal estimate $\xhatw$ [cf.~\eqref{eq.xhat}] and the true value $\x*$ is
\begin{equation}\label{ori_err}
    \mathbf{e}\triangleq \xhatw - \x*.
\end{equation}
We can then formulate the optimal design of $\bw$ as solving the  optimization problem 
\begin{mini}|s|
{\bw\in\mathbb{R}^N}{\mathbb{E} \left\lbrace \lVert \xhatw -  \mathbf{x}^* \lVert_2^2 \right\rbrace, }
{\label{eq.opto}}{} 
\end{mini}
where $\hat{\bbx}$ is defined in \eqref{eq.xhat}.
The inverse relation in $\xhatw$ renders problem~\eqref{eq.opto} challenging to solve in its original form. To overcome this challenge, we first propose two relaxation methods to solve \eqref{eq.opto} for $\bbomega$ leveraging strong (approximate) prior knowledge of the graph signal. Then, we use the \minmax strategy to adapt these two methods to scenarios where only signal bounds are available.

\subsection{Prony's method}\label{sec4.1}

Given the estimate $\xhatw = (\mathbf{I}+\Sw)^{-1}\mathbf{y}$ and the error $\bbe$ in \eqref{ori_err}, a typical approach to design parameters in inverse relationships is to consider Prony's modified error~\cite{hayes}
\begin{equation}\label{mod_err}
\mathbf{e}^\prime \triangleq \mathbf{y}- ( \mathbf{I}+\Sw ) \mathbf{x}^*,
\end{equation}
which is obtained by multiplying both sides of~\eqref{ori_err} by $(\bbI+\Sw)$. Albeit not equivalent to the true error $\bbe$, minimizing the modified error $\mathbf{e}^\prime$ is easier due to the linear relationship in $\bbomega$ and the resulting performance is often satisfactory \cite{hayes}. The NA weights that minimize the error $\bbe^\prime$ in \eqref{mod_err} can be obtained by solving the following problem
\begin{mini}|s|
{\bw\in\mathbb{R}^N}{\mathbb{E} \left\lbrace \lVert \mathbf{y} - ( \mathbf{I}+\Sw ) \mathbf{x}^* \lVert_2^2 \right\rbrace}
{}{}
\addConstraint{\omega_0^* \leq \omega_i^2,\,\text{for } i = 1,\dots,N}\label{eq.probEnProny}
\end{mini}
where the constraint imposes  all entries of $\bbomega$ to satisfy Lemma~\ref{le.res}. It would be more natural to also add the condition on a smaller MSE [cf. \eqref{eq.cond2} and \eqref{eq.corCond}] as a constraint. However, this condition depends on the SNR, i.e., when the SNR is low, it is easier for them to be satisfied; it leads to a smaller set of feasible solutions. Thus, the complication of this condition does not encourage us to add it in the optimization problem. Instead, we observe that Lemma \ref{le.res} already improves the optimal weight design in practice.

The quadratic relation in the optimization variable $\bbomega$ in $\Sw$ [cf. \eqref{eq:s}] makes problem \eqref{eq.probEnProny} non-convex. To obtain $\bbomega$, we follow a two step approach. First we define a positive semi-definite rank-one matrix $\boldsymbol{\Omega}\triangleq\boldsymbol{\omega\omega}^\top$ and solve \eqref{eq.probEnProny} w.r.t the new variable $\boldsymbol{\Omega}$. Then, we find a vector estimate $\hat{\bbomega}$ by performing a rank-one approximation of the obtained matrix~\cite{QCQPrelax}.

Rewriting \eqref{eq.probEnProny} w.r.t. the new variable $\bbOmega$, we obtain
\begin{mini}|s|
{\bW\in\mathcal{S}_+^{N\times N}}{\tr\left\lbrace (\boldsymbol{\Omega}\odot\mathbf{L})^2\mathbf{x}^*\mathbf{x}^{*\top} \right\rbrace}
{\label{eq.prony_relaxed}}{}
\addConstraint{\omega_0^*\leq\Omega_{ii},\,\text{for } i = 1,\dots,N}
\end{mini}
where the derivation of the cost function is reported in Appendix~\ref{appendix.E}. Note that in \eqref{eq.prony_relaxed}, we followed \cite{QCQPrelax} and dropped the non-convex constraint $\text{rank}(\boldsymbol{\Omega})=1$.
Then, problem \eqref{eq.prony_relaxed} becomes a convex semi-definite program (SDP), solvable with off-the-shelf tools \cite{cvx,cvx_theory}, and for instance with interior-point method \cite{vandenberghe1996,helmberg1996}, it has a cost\footnote{We also refer to the recent SDP work \cite{jiang2020} to provide a complete complexity analysis.} of $\ccalO(N^3)$ per iteration. We also observe that the returned solution from~\eqref{eq.prony_relaxed} is consistently a rank-one matrix.
Given then $\bW^*$ from~\eqref{eq.prony_relaxed}, the node-adaptive weight vector $\bw^*$ is equal to the eigenvector of $\bW^*$ with the largest eigenvalue, multiplied with the square root of the eigenvalue. However, more sophisticated rank-one approximations are also possible \cite{QCQPrelax}.

\subsection{Semi-definite relaxation}\label{sec4.2}

Despite its simplicity, Prony's method does not directly relate to the true error $\mathbf{e}$ in~\eqref{ori_err}. Working with the modified error might be viable when the signal-to-noise ratio (SNR) is high but it might lead to a degraded performance when the SNR is low. 
To overcome the latter, we propose here an optimization problem relying on semi-definite relaxation when minimizing the true error in~\eqref{eq.opto}. This approach follows again a two-step procedure: first, we formulate \eqref{eq.opto} w.r.t. the matrix variable $\bbH(\bW) = (\bbI+\bW\odot\bbL)^{-1}$ where $\bW := \bw\bw^\top$, and then we obtain $\bW$ from $\bbH(\bW)$ by means of inversion. Finally, $\bw^*$ is extracted from $\bW^*$ by rank-one approximation \cite{QCQPrelax}.

To start, let us recall the node-adaptive estimate $\xhatw~=~(\mathbf{I}+\bw\bw^\top\odot\bbL)^{-1}\mathbf{y}$ and rewrite~\eqref{eq.opto} w.r.t. $\bbH(\bW)$ as
\begin{mini}|s|
{\bW,\bbH(\bW)\in\mathcal{S}_+^{N\times N}}{\mathbb{E} \left\lbrace \lVert\mathbf{H}(\bW)\mathbf{y} - \mathbf{x}^* \lVert_2^2 \right\rbrace }
{\label{eq.sdp1}}{}
\addConstraint{\mathbf{H}(\bW) = \left( \mathbf{I}+\bW\odot\mathbf{L} \right)^{-1}}
\addConstraint{ \text{rank}(\boldsymbol{\Omega}) = 1}
\addConstraint{\omega_0^*\leq \Omega_{ii},\,\text{for } i = 1,\dots,N.}
\end{mini}
where the last constraint is again Lemma \ref{le.res}. The cost function in~\eqref{eq.sdp1} can be expanded further as (cf. Appendix. \ref{appendix.F})
\begin{align} \label{eq.sdp_cost}
    \mathbb{E} \left\lbrace \lVert\mathbf{H}(\bW)\mathbf{y} - \mathbf{x}^* \lVert_2^2 \right\rbrace &= \tr\{ (\bbH^2(\bW))\\
    &-2\bbH(\bW)+\bbI)\bbx^*\bbx^{*\top} + \bbH^2(\bW)\bbSigma \}\nonumber
\end{align}
which depends on both the signal and noise covariance. Problem~\eqref{eq.sdp1} presents two non-convex constraints: the inverse relationship $\mathbf{H}(\bW) = \left( \mathbf{I}+\bW \odot \bbL \right)^{-1}$ and the rank-one constraint $\text{rank}(\boldsymbol{\Omega}) = 1$. We address the former, by leveraging \emph{semi-definite relaxation} and by writing the constraint in its positive semi-definite convex form~\cite{QCQPrelax}
\begin{equation}\label{mid_step_sdp}
   \mathbf{H}(\bW)  -\left( \mathbf{I}+ \bW \odot \bbL  \right)^{-1} \succeq \mathbf{0}.
\end{equation}
Since $\bbL$, $\bW$, and $\bbI + \bW \odot \bbL$ are positive semi-definite matrices, we can use the Schur complement (see Appendix.\ref{appendix.A}) to reformulate~\eqref{mid_step_sdp} into a convex linear matrix inequality
$$
\begin{pmatrix} 
\mathbf{I}+\bW \odot \bbL & \mathbf{I} \\
\mathbf{I} & \mathbf{H}(\bW) 
\end{pmatrix} \succeq \mathbf{0}.
$$

Regarding the non-convexity of the rank-one constraint, we can relax it as in~\eqref{eq.prony_relaxed}. With these relaxation techniques, we rephrase problem~\eqref{eq.sdp1} into the convex form
\begin{mini}|s|
{\bW,\bbH(\bW)\in\mathcal{S}_+^{N\times N}}{\mathbb{E} \left\lbrace \lVert\mathbf{H}(\bW) \mathbf{y} - \mathbf{x}^* \lVert_2^2 \right\rbrace}
{\label{eq.sdp2}}{}
\addConstraint{\begin{pmatrix} 
\mathbf{I}+ \bW\odot\mathbf{L} & \mathbf{I} \\
\mathbf{I} & \mathbf{H}(\bW) 
\end{pmatrix} \succeq \mathbf{0}}
\addConstraint{\omega_0^*\leq \Omega_{ii},\,\text{for } i = 1,\dots,N}.
\end{mini}
This is also an SDP problem but with more constraints compared to \eqref{eq.prony_relaxed}. After obtaining  $\bbH^*(\bW)$ by solving~\eqref{eq.sdp2}, we can find $\bW$ from the inverse relation below
$$\bbH^*(\bW)(\bbI+\bW\odot\bbL) = (\bbI+\bW\odot\bbL)\bbH^*(\bW) = \bbI$$ by solving the following least squares problem
\begin{mini}|s|
{\bW\in\mathcal{S}_+^{N\times N}}{  \lVert \bbH^*(\bW)(\bbI+\bW\odot\bbL) - \bbI \lVert^2_2 \hspace{5mm} \\
\quad \quad \quad \quad \quad + \quad  \lVert (\bbI+\bW\odot\bbL)\bbH^*(\bW) - \bbI \lVert^2_2.}
{\label{eq.sdp3} }{}
\end{mini}

Given then matrix $\bW^*$ from \eqref{eq.sdp3}, we extract $\bw^*$ by a rank-one approximation \cite{QCQPrelax}. With the obtained NA weights $\bbomega^*$, we subsequently build the filter $\bbH(\bw^*) = (\bbI+\bw^*\bw^{*\top}\odot\bbL)^{-1}$ and obtain the estimate $\hbx(\bw^*)$ as in~\eqref{eq.xhat}. 

The main advantage of the semidefinite relaxation (SDR) approach \eqref{eq.sdp2}-\eqref{eq.sdp3} is that it focuses directly on the true error \eqref{ori_err} rather than the modified error \eqref{mod_err}. However, due to the complexity issue, semidefinite relaxation methods are applicable to medium-sized graphs with up to a thousand  nodes.

\subsection{Min-Max Adaptation}\label{sec4.3}

As it follows from \eqref{eq.probEnProny} and \eqref{eq.sdp1}, both  Prony's method and the SDR method require knowledge of  signal $\bx^*$ to design weights $\bw$. This is possible in a data-driven fashion under the condition that the test and training data have a similar distribution.
In this section, we depart from this assumption and propose a design method that is independent of $\bx^*$ but only requires side information such as signal evolution bounds. The latter is simpler to acquire from a small set of data or by physical considerations.

Consider signal $\bbx^*$ has an evolution bounded in the interval $[ \mathbf{x}_{\rm l},\mathbf{x}_{\rm u} ]$. We can then design the parameter vector $\bw$ for estimator $\hbx(\bw)$ [cf. \eqref{eq.xhat}] as the one that minimizes the MSE of the worst-case scenario, i.e.,

\begin{mini}|s|
{\bw}{\underset{\mathbf{x}^*}{\max}\;\mathbb{E}\left\lbrace \lVert \mathbf{\mathbf{\hat{x}(\boldsymbol{\omega})} -x^*} \lVert_2^2 \right\rbrace}
{\label{eq:minmax1}}{}
\addConstraint{\mathbf{x}_{\rm l} \leq \mathbf{x}^* \leq \mathbf{x}_{\rm u}}
\addConstraint{\omega_0^*\leq \omega_{i}^2,\,\text{for } i = 1,\dots,N.}
\end{mini}
Problems of the form \eqref{eq:minmax1} are known as \minmax  problems. The inner maximization seeks for the signal $\bx^*$ that leads to the worst MSE performance, while the outer minimization finds the parameter $\bw$ that minimizes the worst MSE among all possible choices. 
Different from problems \eqref{eq.probEnProny} and \eqref{eq.sdp1}, signal $\bbx^*$ is now an optimization variable in \eqref{eq:minmax1} and only $\mathbf{x}_{\rm l}$ and $\mathbf{x}_{\rm u}$ are needed. There are efficient methods to solve \minmax problems of the form \eqref{eq:minmax1} such as iterative first-order methods~\cite{minmax} or gradient descent-ascent\cite{gradient_ascent_des}. In the sequel, we detail how \eqref{eq:minmax1} specializes to Prony's method and SDR method.

Following the same rationale as in~\eqref{eq.prony_relaxed}, we can write the \minmax Prony's method as
\begin{mini}|s|
{\bW\in\mathcal{S}^{N\times N}_{+}}{\underset{\bx^*}{\max}\;  \tr\{(\bW\odot\mathbf{L})^2\mathbf{x^*x}^{*\top} \} }
{\label{eq:minmax2}}{}
\addConstraint{\mathbf{x}_{\rm l} \leq \mathbf{x}^* \leq \mathbf{x}_{\rm u}}
\addConstraint{\omega_0^*\leq \Omega_{ii},\,\text{for } i = 1,\dots,N.}
\end{mini}
Since $\bW\odot\bbL$ is positive semidefinite, the cost function is quadratic (convex) over the inner optimization variable $\bbx^*$. Further, it can be shown that the maximizer of the inner problem is at the boundaries, i.e., either $\bx^*=\mathbf{x}_{\rm l}$ or  $\bx^*=\mathbf{x}_{\rm u}$. 
To solve the outer minimization problem w.r.t. $\bW$, we proceed similarly as in \eqref{eq.prony_relaxed}. For the SDR problem \eqref{eq.sdp2} we can also solve the min-max version since the cost function for the latter is also convex in $\bbx^*$. Here again, we optimize the true error [cf. \eqref{ori_err}] at the price of a higher computational complexity.

\section{Numerical results} \label{sec5}

In this section, we compare the performance of the proposed design schemes with state-of-the-art alternatives and illustrate the different trade-offs inherent to NA regularization on synthetic and real-world data from the Molene\footnote{Raw data available at \newline \url{https://donneespubliques. meteofrance.fr/donnees libres/Hackathon/RADOMEH.tar.gz}} and the NOAA\footnote{Raw data available at \newline \url{https://www.ncdc.noaa.gov/data-access/land-based-station-data/land-based-datasets/climate-normals/1981-2010-normals-data}} data sets. We  measure the recovery accuracy between the estimate $\hbx$ and the true signal $\bbx^*$ through the normalized mean squared error, $\text{NMSE} = \lVert\mathbf{\hat{x}-x^*}\lVert_2^2/\lVert\mathbf{x^*}\lVert_2^2$. In these simulations, we used the GSP \cite{gsptoolbox} and CVX \cite{cvx} toolboxes.

\subsection{Synthetic data} \label{sec5.1}
\def\dB{\text{ dB}}
\begin{figure}[t] \vskip-0.5cm
    \centering
    \includegraphics[width=1\linewidth]{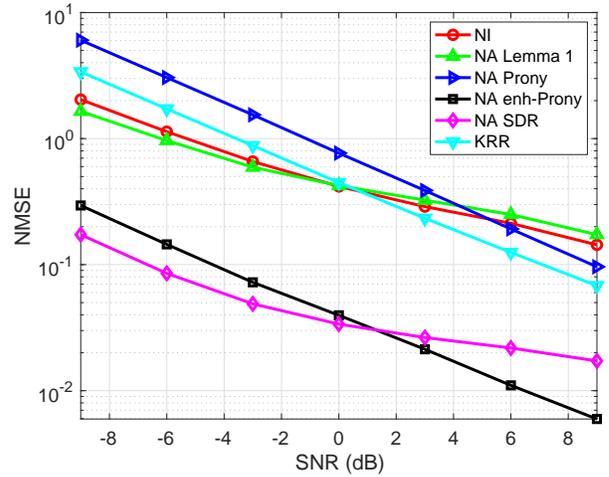}
    \caption{NMSE of different methods as a function of the SNR over an Erdos-Reny\`i graph. The true graph signal is bandlimited to the first $20$ graph frequencies.}
    \label{fig:er}
\end{figure}

In the first set of experiments, we consider synthetic Erdos-Reny\`i graphs of $N = 50$ nodes and link formation probability $0.5$. We generate a synthetic graph signal $\bbx^*$, whose graph Fourier transform is one in the first $20$ coefficients and zero elsewhere. This signal is called a bandlimited graph signal and varies smoothly over the graph \cite{emerging_gsp}; hence, it fits the assumption for Tikhonov regularization \eqref{eq.regls}. We corrupted the signal with a zero-mean Gaussian noise with variance $\sigma_n^2$, so to obtain a signal-to-noise ratio $\text{SNR} = \lVert\mathbf{x^*}\lVert_2^2/(N\sigma_n^2)$. We average the performance over $50$ graphs and $100$ noise realizations leading to a total of $5000$ Monte-Carlo runs. This scenario is also the one used  in \autoref{fig:prop1study}, where we showed the NA regularizer can achieve both a lower variance and MSE.

We first evaluate Prony's method when the true signal is known and the SDR method when both the true signal and noise variance are known. Our rationale is to avoid biases induced by a training set or by focusing solely on the worst-case scenario. We address the latter in the subsequent section with real data. The specific approaches we consider are:
\begin{enumerate}[label={\roman*)},]
\item The benchmark node-invariant regularizer with optimal weight $\omega_0^*=\mathcal{O}(\sqrt{\frac{\theta}{\lambda_2\lambda_N}})$, where $\theta = \sqrt{\frac{1}{\text{SNR}}}$, and $\lambda_2$, $\lambda_N$ are the smallest and the largest non-zero eigenvalues of the graph Laplacian $\mathbf{L}$, respectively \cite{biasvar}.
\item The naive node-adaptive regularization where the NA weights $\bw$ are chosen randomly to satisfy Lemma~\ref{le.res}; i.e., $\omega_i = \sqrt{\omega_0^*}+\omega_0^*\cdot c_i$ for $i=1,\dots,N$, where $c_i$ is uniformly distributed in $[0,1]$.
\item Prony's node-adaptive design, where we do not enforce the constraint of Lemma~\ref{le.res} in problem \eqref{eq.prony_relaxed}.
\item Prony's method [cf. \eqref{eq.prony_relaxed}].
\item The SDR node-adaptive design [cf. \eqref{eq.sdp2}, \eqref{eq.sdp3}].
\item The diffusion kernel ridge regression (KRR) with parameters $\sigma^2_{\text{KRR}} = 1$ and $\mu_{\text{KRR}} = 10^{-4}$, chosen to achieve the best performance \cite{mul_ker2}.
\end{enumerate}
%
\begin{figure}[t] \vskip-0.5cm
    \centering
    \includegraphics[width=1\linewidth]{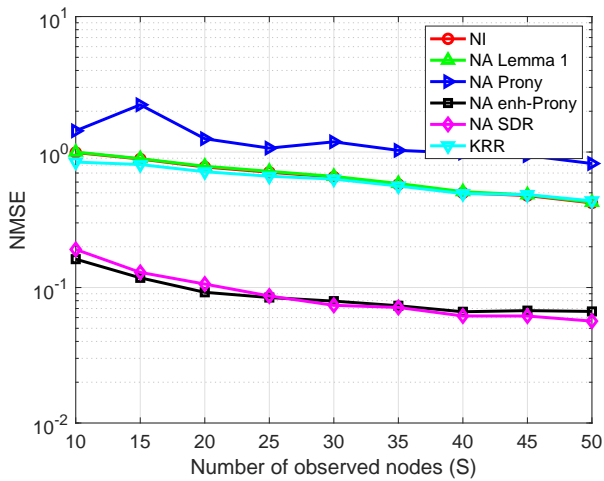}
    \caption{Interpolation performance of the different methods as a function of the number of observed nodes with $\SNR = 0 \text{ dB}$ over an Erdos-Reny\`i graph. The true graph signal is bandlimited in the first $20$ graph frequencies.}
    \label{fig:interp_synth}
\end{figure}

\autoref{fig:er} shows the NMSE recovery performance of the considered methods with respect to the SNR. First, we observe the proposed Prony's method and the SDR reduce the NMSE by one order of magnitude compared with the optimal NI approach and the KRR. As we anticipated in \autoref{fig:prop1study}, even the naive random NA regularizer achieves a comparable performance with these competitive alternatives, ultimately, highlighting the potential of the NA regularizer for graph signal recovery. We also remark the importance of the theoretical result in Lemma~\ref{le.res} in the Prony's problem \eqref{eq.prony_relaxed}, which reduces the MSE by one order of magnitude. We attribute the latter to the fact that Prony's approaches focus on the modified error rather than the true one. This is also evidenced by the comparison with the SDR technique, where Prony's method has a worse NMSE for lower SNRs; i.e., where considering the true error is more effective to deal with the large noise. 

Next, we evaluate the NMSE performance of six different methods for interpolating missing values. We consider noisy observations from the random subset $\ccalM \in \{10,15,20,\dots,50\}$ with an SNR of $0\dB$ to show the robustness of NA regularization. From \autoref{fig:interp_synth}, we observe again the superior performance of Prony's method and the SDR method. On the contrary, the naive weight setting and Prony's unconstrained method offer a similar performance as the benchmark NI regularization. The latter highlights the improvement brought by Lemma \ref{le.res}.

In the sequel, we will conduct experiments over real-world measurements. 
We omit the results of the NA design based on the ground truth signal since they behave the same as for the synthetic data. Instead, we consider the earlier mentioned data-driven and \minmax scenarios.

\subsection{Molene data set}\label{subsec_molene}
\begin{figure*}[h!]\vskip-2.5cm
  \subfloat[Data-driven.]{
	\begin{minipage}[c][1\width]{
	   0.48\textwidth}
	   \centering
	   \includegraphics[width=1\textwidth,]{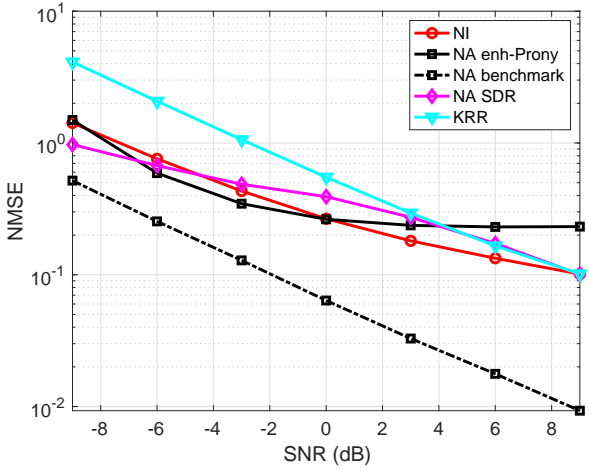}\vskip-1.5cm
	   \label{fig:molene_training}
	\end{minipage}}
 \hfill	
  \subfloat[Min-max.]{
	\begin{minipage}[c][1\width]{
	   0.48\textwidth}
	   \centering
	   \includegraphics[width=1\textwidth]{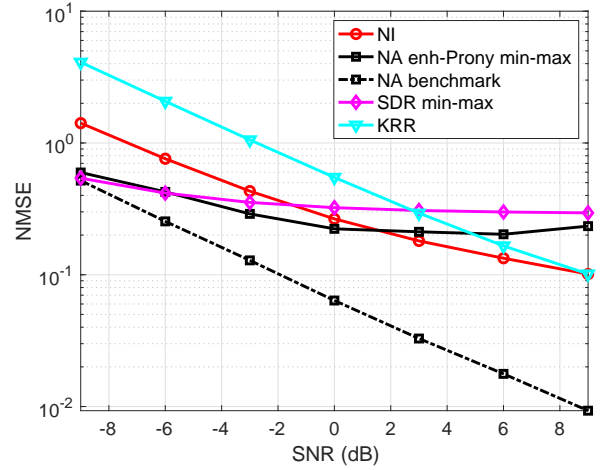}\vskip-1.5cm
	   \label{fig:molene_minimax}
	\end{minipage}}
\caption{NMSE denoising performance of different methods w.r.t SNRs in the Molene data set. (a) Data-driven scenario. Half of the data are used to estimate the node-adaptive parameters. (b) Min-max scenario based on signal evolution bounds. }
\label{fig:molene_main}
\end{figure*}

\begin{figure*}[h!]\vskip-2.5cm
  \subfloat[Data-driven.]{
	\begin{minipage}[c][1\width]{
	   0.48\textwidth}
	   \centering
	   \includegraphics[width=1\textwidth]{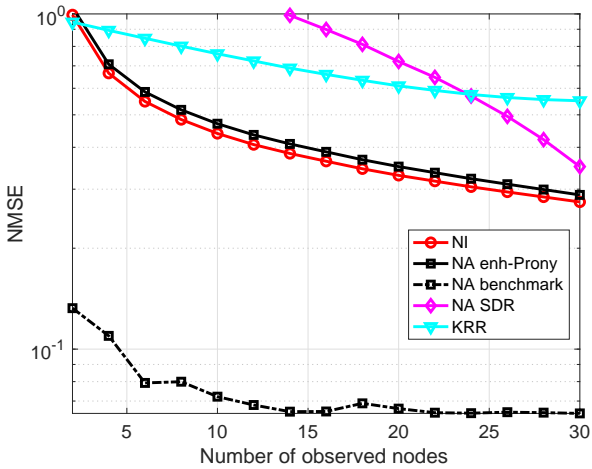}\vskip-1.5cm
	   \label{fig:molene_interp_b}
	\end{minipage}}
 \hfill 	
  \subfloat[Min-max.]{
	\begin{minipage}[c][1\width]{
	   0.48\textwidth}
	   \centering
	   \includegraphics[width=1\textwidth]{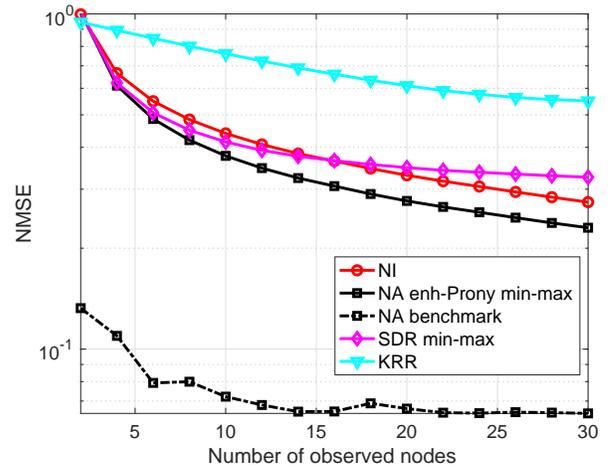}\vskip-1.5cm
	   \label{fig:molene_interp_c}
	\end{minipage}}
 \hfill	
\caption{NMSE interpolation performance of different methods w.r.t the number of sampled nodes on Molene data set. The signal-to-noise ratio is $\text{SNR} = 0 \text{ dB}$. (a) Data-driven scenario. Half of the data are used to estimate the signal covariance matrix. (b) Min-max scenario based on signal evolution bounds.  }
\label{fig:molene_interp}
\end{figure*}
The Molene data set comprises $T = 744$ hourly temperature measurements collected in January $2014$ from $N = 32$ weather stations in the region of Brest, France. We treat each weather station as a node of a graph and build a geometric distance graph in which each node is connected to its five nearest neighbours. The weight of edge $(i,j)$ is $A_{ij} = \text{exp}\{-5 d^2(i,j)\}$, where $d(i,j)$ is the Euclidean distance between stations $i$ and $j$. After removing the mean temperature across space and time, we can view every temporal snapshot as a graph signal. Among the six approaches listed in the former section, we omit the naive Lemma \ref{le.res} based approach as well as Prony's unconstrained method to avoid overcrowded plots since their performance trend is similar to that observed with synthetic data. For the KRR, we tune $\sigma_\text{KRR}^2 = 5$ to reach its best performance. The performance measure is still the NMSE averaged over all the $744$ graph signals. The collected measurements are assumed to be the true signal and we artificially add noise as before. We consider $50$ noise realizations per signal.

To show how different methods behave in the ideal and in a more practical setting, we considered two scenarios. 
First, we considered a data-driven scenario, where half of the temporal data are used to learn the NA weights. Specifically, we use these data to compute an average of $\bbx^*\bbx^{*\top}$ to be  used in Prony's design [cf. \eqref{eq.prony_relaxed}] and the SDR problem [cf.\eqref{eq.sdp2}]. The remaining half graph signals are used for testing. 
As the objective function of~\eqref{eq.sdp2} depends on the noise covariance, for a fixed SNR, the noise level for each signal is different. Thus, we need to design NA weights based on \eqref{eq.sdp2} for each training signal. To avoid doing so, we use one instance from the whole recordings to compute the noise covariance matrix needed in the SDR method, which degrades its performance. 
Second, the min-max method is tested where only the signal lower and upper bounds are known. The bounds are obtained as the lowest and highest temperature records in the dataset. For both scenarios, we consider the NMSE denoising performance as a function of the SNR.

From \autoref{fig:molene_main}, we observe that for more practical scenarios, the performance of the NA approaches degrades by approximately one order of magnitude, where the benchmark is done based on the oracle design \eqref{eq.prony_relaxed} with true signal available. However, they still outperform the KRR counterpart. Specifically, from the results in Figure\autoref{fig:molene_training}, we observe that Prony's method degrades substantially and achieves an NMSE similar to the NI regularizer. Though in the large SNR regime, it gets worse than the NI regularizer, this is because when noise is negligible, the training strategy will result in a suboptimal approximation of $\bx^*\bx^{*\top}$, and thus, suboptimal NA weights in problem \eqref{eq.prony_relaxed}.

Figure\autoref{fig:molene_minimax} shows  the NA algorithms still perform better than competing alternatives in the low SNR regime although the only information is the signal variation range. We attribute the saturation in the high SNR regime to the lack of information needed for designing the NA weights; i.e., the 
NA regularizer will impose a stronger bias on the solution that is not needed to denoise the signal in the high SNR regime. 

We then evaluate the interpolation performance of different methods. We collect noisy observations at  nodes $\ccalM \in \{2,4,6,\dots,28,30\}$  and $\text{SNR} = 0 \dB$. The results are shown in \autoref{fig:molene_interp}.
For the data-driven case in Figure\autoref{fig:molene_interp_b},  Prony's method degrades significantly with a performance worse than the NI regularizer. This could be due to over-fitting the training data. On the other hand, the SDR is not obtaining satisfactory results unless the observations are collected from all of the nodes. This is because the obtained estimate of $\bbx^*\bbx^{*\top}$ does not apply on the SDR method for the interpolation case. In the \minmax scenario in Figure\autoref{fig:molene_interp_c}, we only make use of the signal bounds and obtain a consistently better performance by Prony's method, while the SDR method behaves similar to the NI regularizer.

The performance of the NA regularizer on the Molene data set is in general not surprisingly good. This is because the temperature data is collected geographically in a small region, which inherently leads to a set of smooth signal measurements over the graph. Next, we will test the performance of the NA regularizer over a different data set which is collected from a very large geographical region. The global smoothness then is not intrinsically guaranteed, therefore a superior performance of the NA regularizer is expected. 

\subsection{NOAA data set}
\begin{figure*}[t]\vskip-2.5cm
  \subfloat[Denoising.]{
	\begin{minipage}[c][1\width]{
	   0.48\textwidth}
	   \centering
	   \includegraphics[width=1\textwidth]{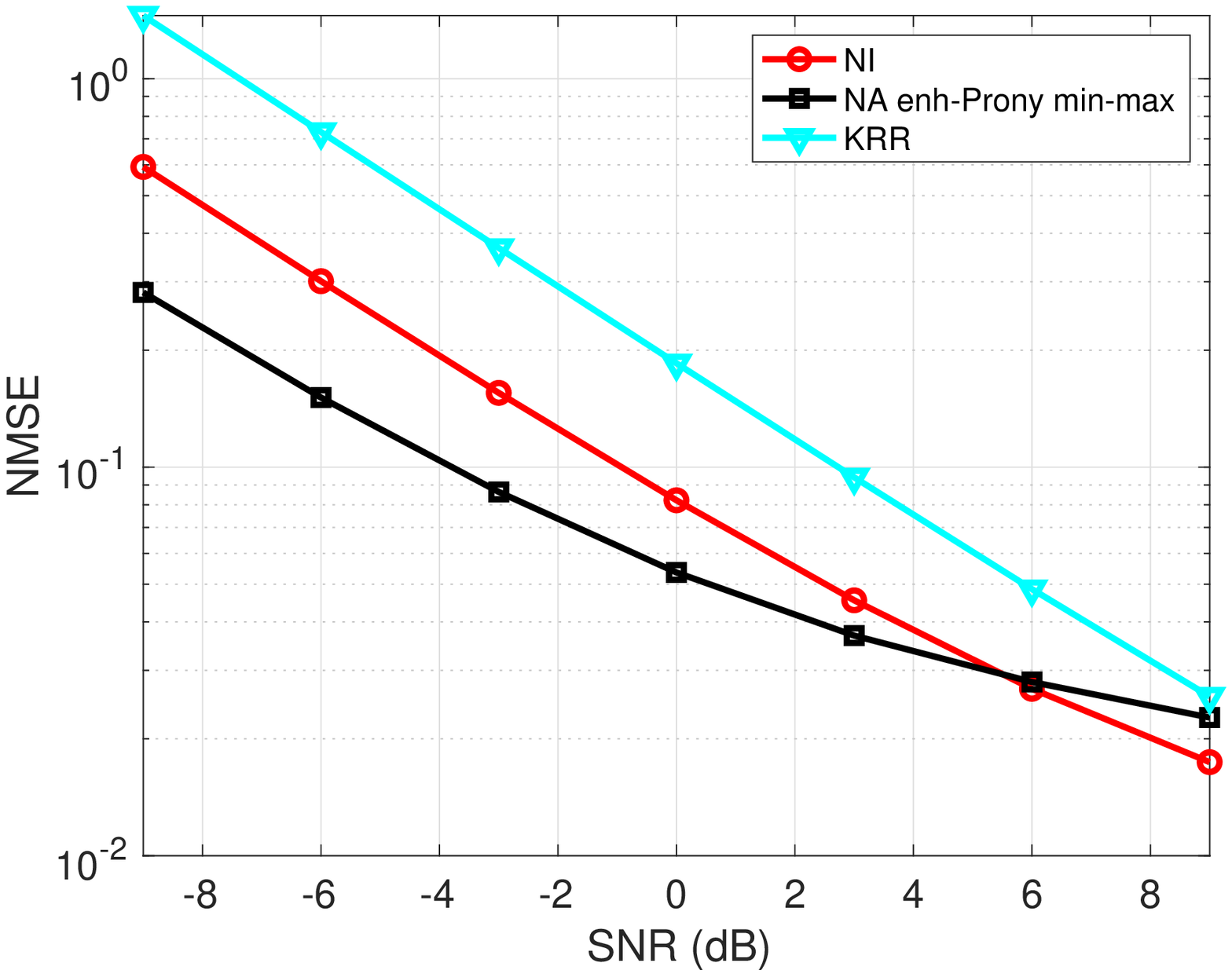}\vskip-1.5cm
	   \label{fig:us_tem_a}
	\end{minipage}}
 \hfill
  \subfloat[Interpolation.]{
	\begin{minipage}[c][1\width]{
	   0.48\textwidth}
	   \centering
	   \includegraphics[width=1\textwidth]{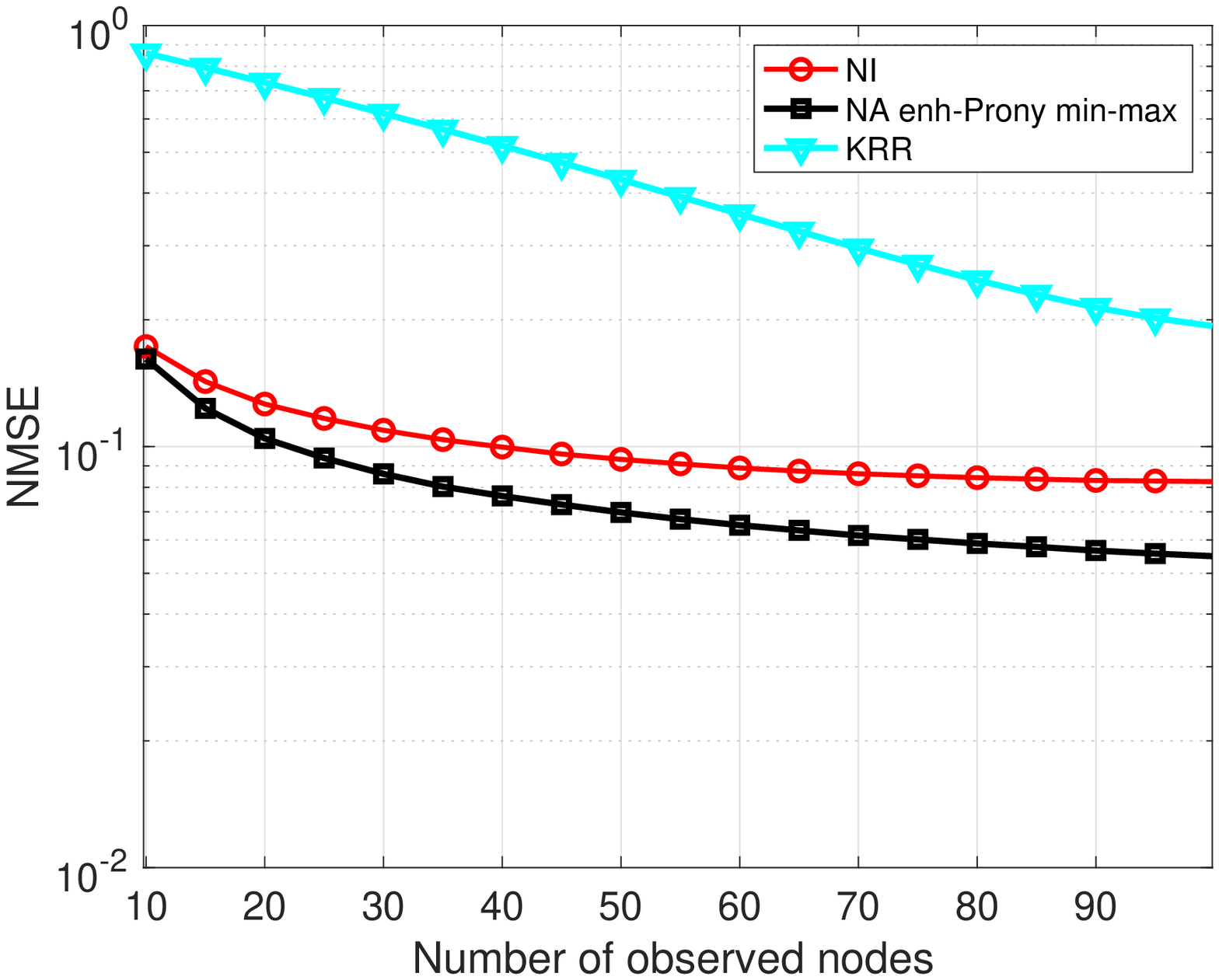}\vskip-1.5cm
	   \label{fig:us_tem_b}
	\end{minipage}}
 \hfill	
\caption{(a) Denoising performance of the enhanced Prony's method w.r.t. SNRs with different methods in \minmax scenario. (b) Interpolation performance w.r.t. SNRs the number of observed nodes, $\SNR=0\dB$}
\label{fig:us_tem}
\end{figure*}

The NOAA data set comprises $T = 8759$ hourly temperature measurements collected across the continental U.S. from $N = 109$ weather stations in $2010$. It is collected from a much larger geographical region compared to the Molene weather data. Following \cite{spg_noaa_data}, we treat each station as a node of a seven nearest neighbor graph based on geographical distances. We measure again the NMSE performance averaged over all the $8759$ signals and $50$ noise realizations per signal. The parameters of all methods are the same as in the former section.

\autoref{fig:us_tem} shows the results for the min-max method applied to denoising and interpolation. From Figure\autoref{fig:us_tem_a}, we observe the improved performance of the proposed NA regularizer over other alternatives. In specific, we see a $3\text{ dB}$ SNR improvement for a fixed NMSE in Prony's method, which is more accentuated at low SNRs. This shows the benefits of the NA regularizer in situations with a high noise level compared with the NI regularizer.

In the interpolation setting, the observed nodes are $\ccalM \in \{10,15,20,\dots,90,95\}$.
From Figure\autoref{fig:us_tem_b}, we observe that the NA regularizer achieves consistently a smaller NMSEs compared with the NI regularizer and KRR method. This improved performance gets more noticeable when the number of observed nodes is larger since the node-adaptive regularizer can better exploit the local signal behavior to find the missing values. In turn, this indicates that when the true signal is not available but only the signal variation bounds are known, i.e., the maximal and minimal signal observations, the NA regularizer can still perform well based on the \minmax strategy.

In practical situations, with proper prior information on the graph signals, we recommend to design the NA regularizer using the \minmax strategy, since it does not require the true signal or training data, but only the signal measurement bounds. From the above results, we see that in harsher situations with a high level of noise present, based on the \minmax strategy, the NA regularizer will behave better than the NI and KRR counterparts in graph signal reconstruction. But when the SNR is high or the noise level is low, we can see that the NI approach is reasonably good. In addition, when the graph signal is globally smooth, such as the Molene weather data, the NA regularizer would behave similar to the NI one. However, if the signal does not have global smoothness, e.g., the NOAA data, the reconstruction ability of the NA regularizer is shown to be much stronger than the NI one.

\section{Conclusions} \label{sec6}
This paper proposed a node-adaptive regularizer for graph signal reconstruction by enhancing the degrees of freedom of the regularization and generalizing the typical global signal smoothness prior. We considered the popular Tikhonov regularizer (also known as graph Laplacian regularizer), and proposed a node-adaptive extension of it. Instead of adopting a shared weight for all nodes, we assign each node a weight to deal with local signal smoothness and possible node heterogeneity. This strategy allows the signal priors to be adapted locally by tuning the regularization parameters, instead of assuming a global smoothness prior. In addition, we evaluated the related bias-variance trade-off, and showed its potential in achieving a lower variance and MSE. Furthermore, we formalized optimal design of the node weights based on Prony's method and semi-definite relaxation. We also proposed a \minmax formulation to deal with situations when no training data is available but only signal side information such as evolution bounds. Finally, numerical results on both synthetic and different type real data were performed to corroborate our findings. 

Meanwhile, in our parallel work \cite{yang2020}, we proposed a first-order iterative method to optimally design the NA weights based on a \minmax game where an energy bound is assumed on the signal measurements. In the future, we expect to further extend the node-adaptive regularization to other regularizers, and aim to reduce the complexity of optimal weight design.

\begin{appendices}
\section{Important lemmas and theorems} \label{appendix.A}
\begin{lemma}
    \textit{Schur complement lemma}: Given any symmetric matrix, 
$
\mathbf{M}=\begin{pmatrix} 
\mathbf{A} & \mathbf{B} \\
\mathbf{B}^\top & \mathbf{C} 
\end{pmatrix}
$, the following conditions are equivalent:
\begin{itemize}
	\item $\mathbf{M}\succeq \bb0$ ($\mathbf{M}$ is positive semi-definite);
	\item $\mathbf{A}\succeq \bb0,\quad (\mathbf{I-AA}^{\dagger})\mathbf{B}=\bb0, \quad \mathbf{C}-\mathbf{B}^\top\mathbf{A}^{\dagger}\mathbf{B}\succeq \bb0$
	\item $\mathbf{C}\succeq \bb0,\quad (\mathbf{I-CC}^{\dagger})\mathbf{B}= \bb0, \quad \mathbf{A}-\mathbf{B}^\top\mathbf{C}^{\dagger}\mathbf{B}\succeq \bb0$
\end{itemize}
\end{lemma}

\begin{lemma} \label{le.trace_equality}
Let $\mathbf{A},\mathbf{B},\mathbf{C}$ be $n\times n$ symmetric matrices, then
\begin{equation}
\begin{aligned}
    & \tr\left( (\mathbf{A}^2-\mathbf{B}^2)\mathbf{C} \right) \\ = & \tr\left( (\mathbf{A}-\mathbf{B})(\mathbf{A}+\mathbf{B})\mathbf{C} \right)-\tr\left((\mathbf{AB-BA})\mathbf{C} \right) \\
     = & \tr\left( (\mathbf{A}-\mathbf{B})(\mathbf{A}+\mathbf{B})\mathbf{C} \right)\\
\end{aligned}
\end{equation}
because $ \tr(\mathbf{ABC}) = \tr\left( (\mathbf{ABC})^\top\right) = \tr(\mathbf{CBA}) = \tr(\mathbf{ACB}). $\footnote{Due to the cyclic property of trace and $\tr(\mathbf{A}) = \tr(\mathbf{A}^\top)$}
\end{lemma}

\begin{lemma} \label{lemma_nsd}
 Let $\bbA$ be an $n \times n$ positive semi-definite matrix, and $\bbB$ an $n \times n$ negative semi-definite matrix, then $\tr(\bbA\bbB) \leq 0$.
\end{lemma}

\begin{proof}
    Consider an eigenvalue-eigenvector equation of matrix $\bbA\bbB$ as follows
    \begin{equation} \label{eq:eigen_eq}
        \bbA\bbB\bbx = \lambda \bbx,
    \end{equation}
    where eigenvalue $\lambda$ is a scalar, and eigenvector $\bbx$ is a vector of length $n$. If we left multiply \eqref{eq:eigen_eq} by $\bbx^\top \bbB$ on both sides, then
    \begin{equation}
        \bbx^\top \bbB \bbA\bbB\bbx = \lambda \bbx^\top \bbB \bbx.
    \end{equation}
    The eigenvalue of $\bbA\bbB$ can be represented as
    \begin{equation}
        \lambda = \frac{\bbx^\top \bbB \bbA\bbB\bbx}{\bbx^\top \bbB \bbx}.
    \end{equation}
    From \cite[Thm. 7.2.7]{matana}, matrix $\bbB\bbA\bbB$ is positive semi-definite independent of matrix $\bbB$. Thus, the numerator $\bbx^\top \bbB\bbA \bbB \bbx $ is nonnegative.
    Since matrix $\bbB$ is negative semi-definite, the denominator $\bbx^\top \bbB \bbx$ is nonpositive. This results into a nonpositive eigenvalue $\lambda$. Thus, the trace of matrix $\bbA\bbB$ is nonpositive.
\end{proof}

\begin{lemma}\label{lemma4}
    Let $n\times n$ matrices $\mathbf{A,B}$ be positive semi-definite, then $\mathbf{A\odot B}$ is positive semi-definite \cite[Thm. 7.5.3]{matana} and
	\begin{equation}
	    \lambda_{\text{max}}(\mathbf{A\odot B})\leq \lambda_{\text{max}}(\mathbf{A})\text{max}\left\lbrace b_{ii}\right\rbrace 
	\end{equation}
	where $b_{ii}$ is the $i$-th diagonal element of $\mathbf{B}$ \cite[7.5.P24]{matana}. Since the inequality is given in the exercise part of \cite{matana}, so we give the proof below.
\end{lemma}

\begin{proof}
	Let $\lambda_{n}$ be the maximal eigenvalue of $\mathbf{A}$, then 
	$\lambda_n\mathbf{I}-\bbA \succeq \mathbf{0}$. Thus, we have $(\lambda_n\mathbf{I}-\bbA)\odot \mathbf{B}$ is positive semi-definite. Let $\mathbf{x}\in\mathbb{R}^n$ be a nonzero vector, then 
	\begin{equation}
		\begin{split}
			& \mathbf{x}^\top \left( (\mathbf{A}-\lambda_n\mathbf{I})\odot \mathbf{B}\right) \mathbf{x} \\ & =\mathbf{x}^\top \left( \mathbf{A}\odot \mathbf{B}\right)\mathbf{x} - \lambda_n\mathbf{x}^\top \left(\mathbf{I}\odot \mathbf{B}\right) \mathbf{x}\leq0,
		\end{split}
	\end{equation}
 which leads to
	\begin{equation}
		\begin{split}
		\mathbf{x}^\top \left( \mathbf{A}\odot \mathbf{B}\right)\mathbf{x} & \leq \lambda_n\mathbf{x}^\top \left(\mathbf{I}\odot \mathbf{B}\right) \mathbf{x} \\ 
		& = \lambda_n\sum_{i=1}^{n}b_{ii}|x_i|^2 \leq \lambda_n\text{max}\left\lbrace b_{ii} \right\rbrace \lVert \mathbf{x}\lVert^2.
		\end{split}
	\end{equation}
	If $\mathbf{x}$ is the eigenvector, the claim holds competing the proof.
\end{proof}
\begin{theorem} \label{theorem2}
	(Weyl's Inequality)\cite[Thm. 4.3.1]{matana} Let $\mathbf{A},\mathbf{B}$ be $n\times n$ Hermitian matrices and let the respective eigenvalues of $\mathbf{A,B}$ and $\mathbf{A+B}$ be $\left\lbrace \lambda_i(\mathbf{A})\right\rbrace^n_{i=1} $, $\left\lbrace \lambda_i(\mathbf{B})\right\rbrace^n_{i=1} $ and $\left\lbrace \lambda_i(\mathbf{A+B})\right\rbrace^n_{i=1} $, each of which is algebraically ordered as $\lambda_{\text{min}}=\lambda_1\leq \lambda_2\leq \dots\leq\lambda_{n-1}\leq\lambda_{n}=\lambda_{\text{max}}$. Then
	\begin{equation}
		\lambda_i(\mathbf{A+B})\leq\lambda_{i+j}(\mathbf{A})+\lambda_{n-j}(\mathbf{B}),\, j=0,1,\dots,n-i
	\end{equation}
	for each $i=1,\dots,n$, with equality for some pair $(i,j)$ if and only if there is a nonzero vector $\mathbf{x}$ such that $\mathbf{Ax}=\lambda_{i+j}(\mathbf{A})\mathbf{x}$, $\mathbf{Bx}=\lambda_{n-j}(\mathbf{B})\mathbf{x}$ and $\mathbf{(A+B)x}=\lambda_{i}(\mathbf{A+B})\mathbf{x}$. Also,
	\begin{equation}
	\lambda_{i-j+1}(\mathbf{A})+\lambda_{j}(\mathbf{B})\leq\lambda_{i}(\mathbf{A+B}),\, j=1,\dots,i
	\end{equation}
	for each $i=1,\dots,n$, with equality for some pair $(i,j)$ if and only if there is a nonzero vector $\mathbf{x}$ such that $\mathbf{Ax}=\lambda_{i-j+1}(\mathbf{A})\mathbf{x}$, $\mathbf{Bx}=\lambda_{j}(\mathbf{B})\mathbf{x}$ and $\mathbf{(A+B)x}=\lambda_{i}(\mathbf{A+B})\mathbf{x}$. If $\mathbf{A}$ and $\mathbf{B}$ have no common eigenvector, then every inequality above is a strict one.
\end{theorem}

\section{Proof of Lemma~\ref{le.res}} \label{appendix.B}
To show  $\text{var}(\boldsymbol{\omega})\leq \text{var}(\omega_0)$, given $\omega_i^2\geq \omega_0, \text{ for } i = 1,2,\dots,n$, it suffices to show  
\begin{equation}
	\begin{split}
		& \quad \text{var}(\boldsymbol{\omega})- \text{var}(\omega_0) =
		\tr(\mathbf{H}^2(\boldsymbol{\omega})\boldsymbol{\Sigma}) - \tr(\mathbf{H}^2(\omega_0)\boldsymbol{\Sigma})\\
		& =  \tr (\left[ (\mathbf{H}(\boldsymbol{\omega}) - \mathbf{H}(\omega_0))(\mathbf{H}(\boldsymbol{\omega}) + \mathbf{H}(\omega_0)) \right] \boldsymbol{\Sigma} ) \leq 0 \\
	\end{split} 
\end{equation}
where the second equality comes from Lemma~\ref{le.trace_equality}.
Since matrices $\boldsymbol{\Sigma}$,  $\mathbf{H}(\boldsymbol{\omega})$ and $ \mathbf{H}(\omega_0)$ are positive semi-definite by definition, the sum of two positive semi-definite matrices is also positive semi-definite. Thus, from Lemma \ref{lemma_nsd}, it suffices to show  $\mathbf{H}(\boldsymbol{\omega}) - \mathbf{H}(\omega_0)\preceq\mathbf{0}$ with the condition $\omega_i^2 > \omega_0$ for $i=1,\dots,N$.

With this condition, we have $\boldsymbol{\omega\omega}^\top\succeq\omega_0\mathbf{11}^\top$. Since $\text{diag}(\boldsymbol{\omega})\mathbf{L}\text{diag}(\boldsymbol{\omega})=\boldsymbol{\omega\omega}^\top\odot\mathbf{L}$ and $\omega_0\mathbf{L}=\omega_0\mathbf{11}^\top\odot\mathbf{L}$, we further have
\begin{equation}
\begin{split}
	(\mathbf{I}+\omega_0\mathbf{11}^\top\odot\mathbf{L}) & - (\mathbf{I}+\boldsymbol{\omega\omega}^\top\odot\mathbf{L}) \\ & =  (\omega_0\mathbf{11}^\top-\boldsymbol{\omega\omega}^\top) \odot \mathbf{L} \preceq \mathbf{0}
\end{split}
\end{equation}
where the equality holds because the Hadamard product is distributive over addition. Then, we left multiply both sides by $	(\mathbf{I}+\omega_0\mathbf{11}^\top\odot\mathbf{L})^{-1}$ and right multiply both sides by $(\mathbf{I}+\boldsymbol{\omega\omega}^\top\odot\mathbf{L})^{-1}$. This does not change the sign because they are both PSD. Hence, we have
\begin{equation}
\begin{split}
	(\mathbf{I}+\boldsymbol{\omega\omega}^\top\odot\mathbf{L})^{-1} &- (\mathbf{I}+\omega_0\mathbf{11}^\top\odot\mathbf{L})^{-1} \preceq \mathbf{0} \\
	\mathbf{H}(\boldsymbol{\omega}) &- \mathbf{H}(\omega_0)\preceq \mathbf{0} \\
\end{split}
\end{equation}
which completes the proof.

\section{Proof of Theorem~\ref{th.mse}} \label{appendix.C}
Proving $\text{mse}(\boldsymbol{\omega}) \leq \text{mse}(\omega_0)$ with the given conditions, is equivalent to showing $\Delta = \text{mse}(\boldsymbol{\omega})- \text{mse}(\omega_0)\leq0 $. We expand the latter as
\begin{equation}
\begin{split}
	\Delta & =  \tr  ( \left( \mathbf{I-H}(\boldsymbol{\omega})\right) ^2\mathbf{x^*x}^{*T}) + \tr(\mathbf{H}^2(\boldsymbol{\omega})\boldsymbol{\Sigma}) \\&  \quad - \tr (\left( \mathbf{I-H}(\omega_0)\right) ^2\mathbf{x^*x}^{*T}) + \tr(\mathbf{H}^2(\omega_0)\boldsymbol{\Sigma}).
\end{split}
\end{equation}
Since $\mathbf{x^*x}^{*T}=\mathbf{P}\boldsymbol{\Sigma}$ holds by definition, by working out the above equation, we have
\begin{multline}
	\Delta  = \tr(
	\boldsymbol{\Sigma}
	\left( 
	\mathbf{H}(\omega_0)-\mathbf{H}(\boldsymbol{\omega})
	\right)\\ \cdot
	\left[ 
	2\mathbf{P}-\left( 
	\mathbf{H}(\boldsymbol{\omega})+\mathbf{H}(\omega_0)
	\right)\left( \mathbf{I}+\mathbf{P}\right) 
	\right] ).
\end{multline}
Due to the covariance matrix $\boldsymbol{\Sigma}\succeq \mathbf{0}$ and from the first condition~\eqref{eq.cond1} we have also that the filter difference is positive semi-definite $ 
\mathbf{H}(\omega_0)-\mathbf{H}(\boldsymbol{\omega})
\succeq \mathbf{0}$. Thus, from Lemma \ref{lemma_nsd}, it suffices to show that 
\begin{equation}
	2\mathbf{P}-\left( 
	\mathbf{H}(\boldsymbol{\omega})+\mathbf{H}(\omega_0)
	\right)\left( \mathbf{I}+\mathbf{P}\right)\preceq \mathbf{0}.
\end{equation}
Further, since $\mathbf{I+P}\succ\mathbf{0}$ and it is invertible, we can focus on proving the smallest eigenvalue is greater than or equal to zero, i.e.,
\begin{equation}
\lambda_{\text{min}}\left\lbrace 
\mathbf{H}(\boldsymbol{\omega})+\mathbf{H}(\omega_0)
-2\mathbf{P}\left( \mathbf{I}+\mathbf{P}\right)^{-1}
\right\rbrace  \geq 0.
\end{equation}
Let us then define the matrix $\boldsymbol{\Gamma}\triangleq \mathbf{P}\left( \mathbf{I}+\mathbf{P}\right)^{-1}$ and scalar $\gamma\in(0,1)$ as its only nonzero eigenvalue to simplify notations. From Theorem~\ref{theorem2}, we have
\begin{multline}
		\lambda_{\text{min}}\left\lbrace 
	\mathbf{H}(\boldsymbol{\omega})+\mathbf{H}(\omega_0)
	- 2\boldsymbol{\Gamma} \right\rbrace 
	 \geq 
	\lambda_{\text{min}}\left\lbrace 
	\mathbf{H}(\boldsymbol{\omega})\right\rbrace  \\
	 +
	\lambda_{\text{min}} \left\lbrace 
	\mathbf{H}(\omega_0) \right\rbrace  
	+	\lambda_{\text{min}} \left\lbrace 
	-2\boldsymbol{\Gamma} \right\rbrace.
\end{multline}
So, a sufficient condition for $\text{mse}(\boldsymbol{\omega}) \leq \text{mse}(\omega_0)$ is
\begin{equation}
	\lambda_{\text{min}}\left\lbrace 
	\mathbf{H}(\boldsymbol{\omega})\right\rbrace 
	+
	\lambda_{\text{min}}\left\lbrace 
	\mathbf{H}(\omega_0)
	\right\rbrace -
	\lambda_{\text{max}}
	\left\lbrace 
	2\boldsymbol{\Gamma}
	\right\rbrace \geq 0
\end{equation}
where $\lambda_{\text{max}}
\left\lbrace 
2\boldsymbol{\Gamma} \right\rbrace = 2\gamma \in (0,2) $,
 and $	\lambda_{\text{min}}\left\lbrace 
\mathbf{H}(\boldsymbol{\omega})\right\rbrace$, $\lambda_{\text{min}}\left\lbrace 
\mathbf{H}(\omega_0)
\right\rbrace$ can be found as follows. From the eigen-decomposition 
\begin{equation}
    \mathbf{H}(\omega_0)\triangleq(\mathbf{I+\omega_0L})^{-1}=\sum_{i=1}^{N}\frac{1}{1+\omega_0\lambda_{i}}\mathbf{u}_i\mathbf{u}_i^\top
\end{equation}
with $\lambda_{i}$ being the $i$-th eigenvalue of the Laplacian $\mathbf{L}$, and $\mathbf{u}_i$ the corresponding eigenvector, we have
\begin{equation}
\lambda_{\text{min}}\left\lbrace 
\mathbf{H}(\omega_0)
\right\rbrace = \lambda_{\text{min}}\left\lbrace 
(\mathbf{I+\omega_0L})^{-1}
\right\rbrace=  \frac{1}{1+\omega_0\lambda_{\text{max}}(\mathbf{L})}.
\end{equation}
From Lemma~\ref{lemma4}, we have
\begin{equation}
    \lambda_{\text{max}}\left\lbrace\boldsymbol{\omega\omega}^\top\odot \mathbf{L}  \right\rbrace \leq \lambda_{\text{max}}(\mathbf{L})\text{max}\left\lbrace \omega_i^2 \right\rbrace
\end{equation}
which results in
\begin{equation}
    \lambda_{\text{max}}\left\lbrace \mathbf{I}+\boldsymbol{\omega\omega}^\top\odot \mathbf{L}  \right\rbrace \leq 1+\lambda_{\text{max}}(\mathbf{L})\text{max}\left\lbrace \omega_i^2 \right\rbrace.
\end{equation}
The lower bounds of $\lambda_{\text{min}}\left\lbrace\mathbf{H}(\boldsymbol{\omega})\right\rbrace$ then follows
\begin{equation}
\lambda_{\text{min}}\left\lbrace
(\mathbf{I}+\boldsymbol{\omega\omega}^\top\odot \mathbf{L})^{-1}  \right\rbrace \geq
\frac{1}{1+\lambda_{\text{max}}(\mathbf{L})\text{max}\left\lbrace \omega_i^2 \right\rbrace }.
\end{equation}
Finally, the sufficient conditions for $\text{mse}(\boldsymbol{\omega}) \leq \text{mse}(\omega_0)$ are
\begin{subequations}
	\begin{align}
		&\omega_i\geq \sqrt{\omega_0}>0, \text{ for } i = 1,2,\dots,n\\
		&\frac{1}{1+\lambda_{\text{max}}(\mathbf{L})\text{max}\left\lbrace \omega_i^2 \right\rbrace } + \frac{1}{1+\omega_0\lambda_{\text{max}}(\mathbf{L})} \geq 2\gamma
	\end{align}
\end{subequations}
which completes the proof.

\section{Proof of Corollary~\ref{cor1}} \label{appendix.D}
If we let the following two hold
\begin{subequations}
\begin{align}
    & \frac{1}{1+\lambda_{\text{max}}(\mathbf{L})\text{max}\left\lbrace \omega_i^2 \right\rbrace } \geq \gamma \\
	& \frac{1}{1+\omega_0\lambda_{\text{max}}(\mathbf{L})} \geq \gamma
\end{align}
\end{subequations}
then we have the condition~\eqref{eq.cond1} sufficiently satisfied. Since $\gamma = \rho/(1+\rho) = 1/(1+\frac{1}{\rho})$, we substitute this relation into above two conditions and have
\begin{subequations}
\begin{align}
		&\text{max}\left\lbrace \omega_i^2 \right\rbrace \lambda_{\text{max}}(\mathbf{L}) \leq \frac{1}{\rho}\\
		&\omega_0\lambda_{\text{max}}(\mathbf{L})\leq \frac{1}{\rho} \label{eq.57b}\\
		&\omega_i\geq \sqrt{\omega_0}>0, \text{ for } i = 1,2,\dots,n.
\end{align}
\end{subequations}
Condition \eqref{eq.57b} can be omitted since $\text{max}\left\lbrace \omega_i^2 \right\rbrace \geq \omega_0$, which completes the proof.

\section{Simplifying the cost function for~\eqref{eq.probEnProny}} \label{appendix.E}
First, we expand the inner term as 
\begin{equation}
\begin{split}
& \mathbb{E} [\lVert \mathbf{y} -(\mathbf{I}+\boldsymbol{\Omega}\odot\mathbf{L})\mathbf{x^*} \lVert_2^2 ] \\ 
= & \mathbb{E} [ \tr\{[\mathbf{y-(I+\Omega\odot L)x^*}][\mathbf{y-(I+\Omega\odot L)x^*]}^\top\}] \\ 
= & \mathbb{E}[
\tr\{\mathbf{yy}^{T} - 2\mathbf{yx}^{*T}(\mathbf{I+\Omega\odot L}) + (\mathbf{I+\Omega\odot L})^2\mathbf{x^*x}^{*T}\}
].\\
\end{split}
\end{equation}
Since the trace operation can be switched with the expectation operation, the cost function is equal to
\begin{equation}
 \tr\{
 \mathbb{E}(\mathbf{yy}^\top) - 2\mathbb{E}(\mathbf{yx}^{*T})(\mathbf{I+\Omega\odot L}) + (\mathbf{I+\Omega\odot L})^2\mathbf{x^*x}^{*T}
\}.
\end{equation}
Further, since $\mathbb{E}(\mathbf{yy}^\top) = \boldsymbol{\Sigma}+\mathbf{x^*x}^{*T}$ and $\mathbb{E}(\mathbf{yx}^{*T}) = \mathbf{x^*x}^{*T}$, the above is equivalent to
\begin{equation}
\begin{split}
 &\tr\{
 \boldsymbol{\Sigma}+\mathbf{x^*x}^{*T} - 2 \mathbf{x^*x}^{*T}(\mathbf{I+\Omega\odot L}) + (\mathbf{I+\Omega\odot L})^2\mathbf{x^*x}^{*T}
\} \\ 
=& \tr\{
 \boldsymbol{\Sigma} + [\mathbf{I-(\mathbf{I+\Omega\odot L})}]^2\mathbf{x^*x}^{*T}
\}\\
=& \tr\{ \boldsymbol{\Sigma} + (\mathbf{\Omega\odot L})^2\mathbf{x^*x}^{*T}\}.
\end{split}
\end{equation}
As our optimization variable is $\bW$, we can drop the unrelated covariance matrix $\boldsymbol{\Sigma}$, which gives the cost function in problem~\eqref{eq.prony_relaxed}.

\section{Detailing the cost function in~\eqref{eq.sdp1}} \label{appendix.F}
The cost function in~\eqref{eq.sdp1} is
\begin{equation}
    \mathbb{E} \left\lbrace \lVert\mathbf{H}(\bW)\mathbf{y} - \mathbf{x}^* \lVert_2^2 \right\rbrace
\end{equation}
which can be further expressed as 
\begin{equation}
    \mathbb{E}{\tr( \bbH^2(\bW)\bby\bby^\top   - 2\bbH(\bW) \bby \bx^{*\top} +\bx^*\bx^{*\top})}.
\end{equation}
Then, exchanging the expectation and trace operator, we obtain the cost function
\begin{equation}
    \tr\{ (\bbH^2(\bW))-2\bbH(\bW)+\bbI)\bbx^*\bbx^{*\top} + \bbH^2(\bW)\bbSigma \}.
\end{equation}
From this expression, we indeed see it includes both the signal and the noise. 

\end{appendices}

\bibliography{ref} 

\bibliographystyle{ieeetr}

\end{document}